\documentclass[11pt,a4paper]{article}

\usepackage{amsthm, amsfonts,amssymb,euscript}
\usepackage{amsmath}
\usepackage{bm}
\usepackage{amssymb}
\usepackage{amsthm}
\usepackage{graphicx}
\usepackage{caption}
\usepackage{subcaption}
\usepackage{amsfonts}
\usepackage{float}
\usepackage{color}
\usepackage{slashed}
\usepackage[affil-it]{authblk}

\usepackage[usenames,dvipsnames,svgnames,table]{xcolor}
\usepackage[colorlinks=true]{hyperref}
\hypersetup{linkcolor=BrickRed, urlcolor=green, citecolor=blue, linktoc=page}

\numberwithin{equation}{section}

\newtheorem*{proposition*}{Proposition}
\newtheorem*{theorem*}{Theorem}
\newtheorem*{conjecture*}{Conjecture}
\newtheorem*{claim*}{Claim}
\newtheorem*{lemma*}{Lemma}
\newtheorem*{corollary*}{Corollary}
\newtheorem{theorem}{Theorem}[section]
\newtheorem{proposition}[theorem]{Proposition}

\newtheorem{lemma}[theorem]{Lemma}

\newtheorem*{definition*}{Definition}

\newtheorem*{assumption*}{\mathcal{A}ssumption}

\newtheorem*{remark*}{Remark}
\newtheorem{remark}{Remark}[section]

\newcommand{\R}{\mathbb{R}}

\allowdisplaybreaks

\begin{document}

\markboth{Yannis Angelopoulos, Stefanos Aretakis and Dejan Gajic}
{Logarithmic corrections in the asymptotic expansion for the radiation field along null infinity}

%
%

\title{Logarithmic corrections in the asymptotic expansion for the radiation field along null infinity}

\author[1]{Yannis Angelopoulos \thanks {yannis@math.ucla.edu}}
\author[2]{Stefanos Aretakis\thanks {aretakis@math.toronto.edu}}
\author[3]{Dejan Gajic \thanks {D.Gajic@dpmms.cam.ac.uk}}
	\affil[1]{\small Department of Mathematics, University of California, Los Angeles, CA 90095, United States}
	\affil[2]{\small Department of Mathematics, University of Toronto, 40 St George Street, Toronto, ON, Canada}
	\affil[3]{\small Centre for Mathematical Sciences, University of Cambridge, Wilberforce Road, Cambridge CB3 0WB, UK}
	\date{}




\maketitle


\begin{abstract}
{\bfseries Abstract.}\quad 
We obtain the second-order late-time asymptotics for the radiation field of solutions to the wave equation on spherically symmetric and asymptotically flat backgrounds including the Schwarzschild and sub-extremal Reissner--Nordstr\"{o}m families of black hole spacetimes. 
These terms appear as logarithmic corrections to the leading-order asymptotic terms which were rigorously derived in our previous work. 
Such corrections have been heuristically and numerically derived in the physics literature in the case of a non-vanishing Newman--Penrose constant. In this case, our results provide a rigorous confirmation of the existence of these corrections. On the other hand, the precise logarithmic corrections for spherically symmetric compactly supported initial data (and hence, with a vanishing Newman--Penrose constant) explicitly obtained here appear to be new.

\end{abstract}


\section {Introduction}

\subsection{Introduction and background}
\label{sec:IntroductionAndBackground}

This paper obtains second-order late-time asymptotics for spherically symmetric solutions to the linear wave equation 
\begin{equation}
\Box_g\psi=0
\label{we}
\end{equation}
on a class of 4-dimensional spherically symmetric, stationary and asymptotically flat spacetimes $(\mathcal{M},g)$ which include Schwarzschild and sub-extremal Reissner--Nordstr\"{o}m backgrounds as special cases. Higher-order late-time asymptotics are relevant to the following problems in general relativity: 
\begin{enumerate}
	\item strong cosmic censorship, \item propagation of gravitational waves, \item  stability problems (e.g. for black holes), \item  the long-time behavior of the Einstein equations.  

\end{enumerate}

The seminal heuristic work of Price \cite{price1972} from 1972 suggests that solutions $\psi$ to the wave equation on Schwarzschild backgrounds $(\mathcal{M}_M,g_M), M>0$ arising from smooth compactly supported initial data satisfy the following late-time polynomial law
\begin{equation}
\psi(t,r_0,\theta,\varphi)\sim \frac{1}{t^3}
\label{price}
\end{equation}
as $t\rightarrow \infty$ along constant $r=r_0>2M$ hypersurfaces. Subsequently, Leaver \cite{leaver}  and Gundlach--Price--Pullin \cite{CGRPJP94b} suggested the following asymptotic behavior for the radiation field $r\psi$ along the future null infinity $\mathcal{I}^{+}=\{ (u,r=\infty,\theta,\phi): u\in\mathbb{R}, (\theta,\varphi)\in\mathbb{S}^2\}$:
\begin{equation}
r\psi|_{\mathcal{I}^{+}}\sim \frac{1}{u^2},
\label{pullin}
\end{equation} 
as the retarded time $u\rightarrow +\infty$. A rigorous proof of \eqref{price}, \eqref{pullin} was obtained in \cite{paper2} where it was in fact shown that there are  uniform constants $C,C_R$ and initial data norms $E[\psi]$ such that 
\begin{equation}
\left| \psi(t,r,\theta,\varphi)+8I^{(1)}_{0}[\psi]\cdot \frac{1}{t^3}\right|\leq C_R\sqrt{ E[\psi]}\cdot \frac{1}{t^{3+\epsilon}}
\label{a1}
\end{equation} for all $r\leq R$ 
and
\begin{equation}
\left| r\psi|_{\mathcal{I}^{+}}(u,\theta,\varphi)+2I^{(1)}_{0}[\psi]\cdot \frac{1}{u^2}\right|\leq C\sqrt{ E[\psi]}\cdot \frac{1}{u^{2+\epsilon}}
\label{a2}
\end{equation}
along the null infinity $\mathcal{I}^{+}, $
where $\epsilon$ is an appropriately small, positive number. The constant $I^{(1)}_{0}[\psi]$ is explicitly given in terms of the initial data of $\psi$ on a Cauchy hypersurface $\Sigma_0$ by \eqref{npfcompact}.
This work provided a mathematically rigorous derivation of sharp upper and lower pointwise bounds for the evolution of the scalar fields on Schwarzschild backgrounds. In fact, the estimates \eqref{a1} and \eqref{a2} provide the precise leading-order late-time asymptotics for $\psi$ with a coefficient that is explicitly given in terms of initial data. 

For initial data that are not compactly supported, the leading-order late-time asymptotic estimates take the form
\begin{equation}
\left| \psi(t,r,\theta,\varphi)-4I_{0}[\psi]\cdot \frac{1}{t^2}\right|\leq C_R\sqrt{ E[\psi]}\cdot \frac{1}{t^{2+\epsilon}}
\label{a12}
\end{equation} for all $r\leq R$ 
and
\begin{equation}
\left| r\psi|_{\mathcal{I}^{+}}(u,\theta,\varphi)-2I_{0}[\psi]\cdot \frac{1}{u}\right|\leq C\sqrt{ E[\psi]}\cdot \frac{1}{u^{1+\epsilon}}.
\label{a22}
\end{equation}
Here $I_{0}[\psi]$ denotes the \emph{Newman--Penrose constant} for $\psi$, discovered originally in \cite{np2} (see Section \ref{np}).

The derivation of \eqref{a1}--\eqref{a22} uses in particular the improved upper bound decay estimates for the wave equation obtained in \cite{paper1,paper2}.

For previous upper bounds for the wave equation on black hole backgrounds we refer the reader to \cite{lecturesMD, Dafermos2016, MDIR05, part3, dhr-teukolsky-kerr, tataru3, metal, blukerr, moschidis1,  volker1, peter2, semyon1, other1, dssprice, kro} and references therein. Previous lower bounds were rigorously obtained in \cite{luk2015}. We remark that the improved decay estimates for the wave equation obtained in \cite{paper1} also play a role in the recent work of Klainerman--Szeftel on the stability of Schwarzschild backgrounds under axially symmetric polarized perturbations \cite{klainerman17}.

In this paper, we derive the second-order late-time asymptotics for the radiation field of spherically symmetric solutions to the wave equation on Schwarzschild backgrounds (see Section \ref{thegeometricassumptions} for the precise class of admissible spacetimes). These asymptotics, in the case of non-compactly supported initial data, were expected to contain logarithmic corrections to \eqref{a22} in the form\footnote{As an aside, we mention that Christodoulou \cite{nopeeling} obtained logarithmic corrections for the expansion in $r$ of the gravitational field as we approach null infinity (that is in the limit $r\rightarrow \infty$). See also the recent \cite{l16}.} of $\frac{\log u}{u^2}$. Indeed, in \cite{Gomez1994} G\'{o}mez--Winicour--Schmidt studied a spherically symmetric massless scalar field in the region $\left\{r > R\right\}$ of Schwarzschild spacetimes, with non-compactly supported initial data and reflecting boundary conditions imposed on the surface $r = R > 2M$. It was \emph{perturbatively} argued that the following asymptotic behavior for the radiation field $\psi$ should hold along the future null infinity $\mathcal{I}^+$:
\begin{equation}
\left.r\psi\right|_{\mathcal{I^{+}}}(u,\cdot)=2I_0[\psi]\cdot \frac{1}{u}-4MI_{0}[\psi]\cdot \frac{\log u}{u^2}+O(u^{-2}).
\label{gomez}
\end{equation}

Similar logarithmic corrections were derived \emph{numerically} for the asymptotic expansion of massless spherically symmetric scalar fields on extremal Reissner--Nordstr\"{o}m by Lucietti--Murata--Reall--Tanahashi in \cite{hm2012}. Specifically, the following expansion was derived along the future event horizon $\mathcal{H}^+$:
\begin{equation}
\left.\psi\right|_{\mathcal{\mathcal{H}^{+}}}(v,\cdot)=-2H_0[\psi]\cdot \frac{1}{v}+4H_{0}[\psi]\cdot \frac{\log v}{v^2}+O(v^{-2}),
\label{harvey}
\end{equation}
where $H_{0}[\psi]$ denotes the conserved charge on the extremal horizon (see \cite{aretakis1,aretakis2,aretakis3,aretakis4,hj2012,aretakisglue,zimmerman1}).  It is important to note that the result in \cite{hm2012} is consistent with that of \cite{Gomez1994} via the Couch--Torrence conformal symmetry on extremal Reissner--Nordstr\"{o}m (see also \cite{godazgar17}). We further remark that the precise expansion \eqref{harvey} plays a crucial role in the behavior of the scalar fields in the interior of extremal black holes and, in particular, in the regularity properties of the scalar field at the inner horizon (see \cite{gajic,gajic2,harvey2013}).

For the \emph{existence} of a polyhomogeneous asymptotic expansion of radiation fields along future null infinity in a class of asymptotically flat Lorentzian manifolds without trapped null geodesics, we refer to \cite{baskinw, baskinwang, baskin16}.

{In this paper, \textit{we give a rigorous derivation of the second-order late-time asymptotics of the radiation field $r\psi |_{\mathcal{I}^+}$ in the retarded time $u$ for a class of asymptotically flat, stationary and spherically symmetric spacetimes that includes the sub-extremal Reissner--Nordstr\"om and Schwarzschild families}.} We moreover find the precise dependence of the constants appearing in these asymptotics in terms of initial data for $\psi$ on a Cauchy hypersurface $\Sigma_{0}$ which, in the black hole case, crosses the future event horizon to the future of the bifurcation sphere and terminates at future null infinity. We consider two cases: in the first case, we take the initial data to have a finite \emph{but non-zero} Newman--Penrose constant $I_0$ and in the second case, we consider compactly supported \emph{and spherically symmetric} initial data (which implies that $I_0$ vanishes). See Section \ref{sec:TheMainResults}. To our knowledge, there have been no results in the literature, either rigorous, heuristic or numerical, for the precise higher-order late-time asymptotics of the radiation field at null infinity in the context of the wave equation on asymptotically flat black hole backgrounds \textit{with compactly supported initial data}. See, however, \cite{casott1,casott2, casott3} for heuristic results on higher-order logarithmic corrections in the late-time asymptotics in Schwarzschild and Kerr at \emph{finite radius}.

The symmetry assumption on the compactly supported initial data will be removed in an upcoming work, where it will be shown that the late-time behavior for the radiation field of the \textit{non-spherically symmetric projection} $\psi-\frac{1}{4\pi}\int_{\mathbb{S}^2}\psi$ is $O(u^{-3})$ along future null infinity. Hence, the estimate \eqref{eq:2ndasympphiNP0infty} is expected to provide the second-order asymptotics for general solutions without any symmetry assumptions.

Our method expands on physical space techniques introduced in \cite{paper1,paper2}.  In particular, we avoid explicit representations of solutions to the wave equation, the use of Fourier transform or the use of conformal compactifications. The results in particular apply to asymptotically flat Lorentzian manifolds with hyperbolic trapping such as the Schwarzschild backgrounds. See Section \ref{sec:TheMainResults} for the detailed statements of the theorems.

In future works, we will obtain higher-order asymptotics for scalar fields with non-spherically symmetric initial data and we will also address the higher-order asymptotics in extremal black holes. 

\subsection{The main results}
\label{sec:TheMainResults}

We consider 4-dimensional spacetimes $(\mathcal{M},g)$ as defined in Section \ref{thegeometricassumptions}. In particular, the metric $g$ is given by
\[g=-D(r)dudv+r^2(d\theta^2+\sin^2\theta d\varphi^2)\]
 in double null coordinates $(u,v,\theta,\varphi)$. Here $r=r(u,v)$ is the area-radius of the spheres of symmetry.  The metric component $D(r)$ satisfies \eqref{dbehavior}.
We remark that the Schwarzschild and sub-extremal Reissner--Nordstr\"{o}m backgrounds satisfy these assumptions. 

The trace of a  function $f(u,v,\theta,\varphi)$ on null infinity is defined  as follows
\[f|_{\mathcal{I}^+}(u,\theta,\varphi):= \lim_{v\rightarrow \infty}f(u,v,\theta,\varphi).\]

 The stationary Killing field $T$ is equal to 
\[T=\partial_v+\partial_u.\] 
We consider the spacelike-null hypersurfaces $\Sigma_{\tau}$, with $\tau\geq 0$, as defined in Section \ref{thegeometricassumptions} which, in the black hole case, cross the future event horizon and terminate at  future null infinity. 

We further consider regular initial data for the wave equation on the hypersurface $\Sigma_0$ and require that the arising solutions $\psi $ satisfy the assumptions of Section \ref{ivp}. 

The Newman--Penrose constant $I_{0}[\psi]$ is given in Section \ref{np} while  the time-integral $\psi^{(1)}$ along with the time-inverted Newman--Penrose constant $I_{0}^{(1)}[\psi]$ are given in Section \ref{timeinverted}.

Finally, the weighted initial data norms
\[ E[\psi],\ \ E_{T}[\psi],\ \ P[\psi],\ \ P_{T}[\psi] \]
are defined in \ref{energynorms}.

Throughout the paper we will apply ``big O'' notation. Let $\beta \in \R$. Then we denote with $O_k(r^{\beta})$ a $C^k$ function $f : [r_{\rm min}, \infty) \rightarrow \R$ that satisfies the following property: for all $0 \leq  j \leq k$, there exist uniform constants $C_j > 0$, such that $$\left|\frac{d^jf}{dr^j}\right|(r) \leq  C_j r^{-\beta-j}.$$ Furthermore, we denote with $O((v-u-1)^{-\beta})$, $O((u+1)^{-\beta})$ and $O(v^{-\beta})$ spherically symmetric $C^0$ functions $f: \{r\geq R\} \to \R$ that satisfy, respectively, the following properties: there exists a uniform constant $C > 0$, such that \begin{align*}|f(u,v)| \leq &\: C(v-u-1)^{-\beta},\\ |f(u,v)| \leq &\: C(u+1)^{-\beta}, \\ |f(u,v)| \leq &\: Cv^{-\beta}.\end{align*}

The next theorem establishes the second-order asymptotics for the radiation field along null infinity for compactly supported initial data (and more generally for data with vanishing Newman--Penrose constant $I_0[\psi]=0$).
\begin{theorem}\label{thm1}\textbf{(Second-order asymptotics for $r\psi|_{\mathcal{I}^+}$ with $I_{0}[\psi]=0$)} For all spherically symmetric solutions $\psi$ to the wave equation \eqref{we} on the spacetimes $(\mathcal{M},g)$ given in Section \ref{thegeometricassumptions} with compactly supported initial data, there exists a constant $C>0$ that depends only on $D(r)$ and $\Sigma_0$ such that the following asymptotic estimate holds along  the future null infinity $\mathcal{I}^+\cap\{u\geq 0\}$:
\begin{equation}
\label{eq:2ndasympphiNP0infty}
\begin{split}
&\left|r\psi|_{\mathcal{I}^+}(u)+2I_0^{(1)}[\psi]\cdot\frac{1}{(u+1)^{2}}-8MI_0^{(1)}[\psi]\cdot \frac{\log(u+1)}{(u+1)^{3}}\right|\\
& \ \ \ \ \ \ \ \ \  \ \ \ \ \  \  \leq C\left(I_0^{(1)}[\psi]+P_T[\psi^{(1)}]+\sqrt{E_T[\psi^{(1)}]}\right)\cdot \frac{1}{(u+1)^{3}}.
\end{split}
\end{equation}
The constant $I_0^{(1)}[\psi]$ is explicitly given in terms of the initial data of $\psi$ on $\Sigma_{0}$ by \eqref{npfcompact}. 
\end{theorem}
\begin{remark}For general solutions $\psi$ to the wave equation we expect that the leading-order late-time asymptotics for the radiation field of the non-spherically symmetric projection $\psi-\frac{1}{4\pi}\int_{\mathbb{S}^2}\psi$ is $O(u^{-3})$ along $\mathcal{I}^+$. Hence, the estimate \eqref{eq:2ndasympphiNP0infty} is expected to provide the second-order asymptotics for general solutions without any symmetry assumptions. This is addressed in an upcoming work. 
\end{remark}
Theorem \ref{thm1} is proved in Section \ref{pthm1}. It is important to remark that the norms $P_T[\psi^{(1)}], {E_T[\psi^{(1)}]}$ are finite and can in fact be bounded by appropriate norms of the initial data of $\psi$ (see \cite{paper2}). 

The next theorem derives the second-order asymptotics for the radiation field along null infinity in the case when the Newman--Penrose constant $I_0[\psi]$ is non-vanishing. In this case, we drop the symmetry assumption and we can in fact consider general solutions $\psi$ to the wave equation. 

\begin{theorem}\label{thm2} \textbf{(Second-order asymptotics for $r\psi|_{\mathcal{I}^+}$ with $I_{0}[\psi]\neq 0$)} For all solutions $\psi$ to the wave equation \eqref{we} on the spacetimes $(\mathcal{M},g)$ given in Section \ref{thegeometricassumptions} with non-vanishing Newman--Penrose constant $I_0[\psi]$, there exists a constant $C>0$ that depends only on $D(r)$ and $\Sigma_0$ such that the following asymptotic estimate holds along the future null infinity $\mathcal{I}^+\cap\{u\geq 0\}$:
\begin{equation}
\label{eq:asymphiinfty}
\begin{split}
&\Bigg|r\psi|_{\mathcal{I}^+}(u)-2I_0[\psi]\cdot \frac{1}{u+1}+4MI_0[\psi]\cdot\frac{\log (u+1)}{(u+1)^{2}}\Bigg|\\&\ \ \ \ \ \ \ \ \ \ \ \ \ \ \ \ \ \ \ \ \ \ \ \leq  C\left(I_0[\psi]+\sqrt{E[\psi]}+P[\psi]\right)\cdot \frac{1}{(u+1)^{2}}.
\end{split}
\end{equation}
\end{theorem}

Finally, the third theorem derives the second-order asymptotics for the radiation field of the $T$-derivative $\psi$ along null infinity in the case when the Newman--Penrose constant $I_0[\psi]$ is non-vanishing.  
\begin{theorem}\label{thm3}  \textbf{(Second-order asymptotics for $T(r\psi)|_{\mathcal{I}^+}$ with $I_{0}[\psi]\neq 0$)} For all solutions $\psi$ to the wave equation \eqref{we} on the spacetimes $(\mathcal{M},g)$ given in Section \ref{thegeometricassumptions} with non-vanishing Newman--Penrose constant $I_0[\psi]$, there exists a constant $C>0$ that depends only on $D(r)$ and $\Sigma_0$ such that the following asymptotic estimate holds along the future null infinity $\mathcal{I}^+\cap\{u\geq 0\}$:
\begin{equation}
\label{eq:2ndasympTphiinfty}
\begin{split}
&\left|T(r\psi)|_{\mathcal{I}^+}(u)+2I_0[\psi]\cdot \frac{1}{(u+1)^{2}}-8MI_0[\psi]\cdot \frac{\log(u+1)}{(u+1)^{3}}\right|\\
&\ \ \ \ \ \ \ \ \ \ \ \ \ \ \ \ \ \ \ \ \ \ \ \ \ \ \ \ \ \ \ \leq  C\left(I_0[\psi]+P_T[\psi]+\sqrt{E_T[\psi]}\right)\cdot\frac{1}{(u+1)^{3}}.
\end{split}
\end{equation}
\end{theorem}
The Theorems \ref{thm2}, \ref{thm3} are proved in Sections \ref{pthm2}, \ref{pthm3}, respectively. We note that Theorem \ref{thm1} is proved using Theorem \ref{thm3} and the time-inversion construction introduced in \cite{paper2}.

\subsection{Discussion on the full asymptotic expansion}
\label{sec:Discussion}

In this section we provide a brief discussion on the coefficients in the \emph{full} asymptotic expansion of  scalar fields $\psi$ and their radiation fields $r\psi$. We will not produce here the detailed full asymptotic expansion, but rather we will provide implicit expressions for all quantities of the initial data of $\psi$ which determine its late-time asymptotic expansion to all orders. 

For the purposes of this discussion, we consider spherically symmetric\footnote{Non-spherically symmetric solutions will be addressed in an upcoming paper. } solutions $\psi$ to the wave equation with smooth, compactly supported initial data on a Cauchy hypersurface $\Sigma_0$ that crosses the event horizon and terminates at null infinity.

We first consider a \textbf{fixed} smooth spherically symmetric solution $\boldsymbol{\psi}_{\text{seed}}$ to the wave equation with Newman--Penrose constant equal to: 
\[ I_{0}^{(0)}[\boldsymbol{\psi}_{\text{seed}}]=1.\]

We moreover assume that with respect to the coordinate $\rho=\frac{1}{r}$ along $\Sigma_0$, $\boldsymbol{\psi}_{\text{seed}}$ is smooth at $\rho=0$ (which corresponds to $r\to \infty$).
Clearly, for this solution we also have
\begin{equation}
I_{0}^{(0)}[\boldsymbol{\psi}_\text{seed}]=I_{0}^{(1)}[T\boldsymbol{\psi}_\text{seed}]=I_{0}^{(2)}[T^2\boldsymbol{\psi}_\text{seed}]=\cdots =I_{0}^{(n)}[T^n \boldsymbol{\psi}_\text{seed}]=1
\label{tkpsiseed}
\end{equation}
for all $n\geq 1$. 

Let's fix $k\in\mathbb{N}$. We consider an arbitrary smooth, spherically symmetric solution $\psi$ to the wave equation with compactly supported initial data. From the results of \cite{paper2}, the leading-order term in the asymptotic expansion of 
 $\psi|_{r=R}$ is $-8I_{0}^{(1)}[\psi] \cdot \frac{1}{\tau^{3}}$ and the leading-order term in the asymptotic expansion of 
 $r\psi|_{\mathcal{I}^{+}}$ is $-2I_{0}^{(1)}[\psi] \cdot \frac{1}{\tau^{2}}$. Hence, the coefficient of the leading order terms is proportional only to $I_{0}^{(1)}[\psi]$. We will next find all quantities of the initial data of $\psi$ which produce terms in the asymptotic expansions of $\psi|_{r=R}$ and $r\psi|_{\mathcal{I}^{+}}$ that decay slower than $\frac{1}{\tau^{k+3+\epsilon}}$ for any $\epsilon>0$ (and any $k\geq 1$). 

We define the following auxiliary solutions to the wave equation:
\begin{equation*}
\begin{split}
\psi_0:=& \psi,\\
\psi_1 :=& \psi_0-I_{0}^{(1)}[\psi_0]\cdot T \boldsymbol{\psi}_{\text{seed}}, \\
\psi_2 :=& \psi_1-I_{0}^{(2)}[\psi_1]\cdot T^2 \boldsymbol{\psi}_{\text{seed}}, \\
\cdots \ & \\
\psi_k :=& \psi_{k-1}-I_{0}^{(k)}[\psi_{k-1}]\cdot T^k \boldsymbol{\psi}_{\text{seed}}. \\
\end{split}
\end{equation*}
Note that the above recursive construction of the auxiliary solutions $\psi_i$ is well-defined since for every $i\geq 1$ we have 
\[I_{0}^{(0)}[\psi_i]=I_{0}^{(1)}[\psi_i]=I_{0}^{(2)}[\psi_i]=\cdots =I_{0}^{(i)}[\psi_i]=0 \]
and hence $I_{0}^{(i+1)}[\psi_i]$ is well-defined. Since
\[I_{0}^{(0)}[{\psi_k}]=I_{0}^{(1)}[{\psi_k}]=I_{0}^{(2)}[{\psi_k}]=\cdots = I_{0}^{(k)}[{\psi_k}]=0\] 
we obtain by the Theorem 1.5 of \cite{paper2}:
\begin{equation}
\psi_k|_{\{r=R\}}=C_k\cdot  I_{0}^{(k+1)}[\psi_k]\cdot \frac{1}{\tau^{3+k}}+O\left(\frac{1}{\tau^{3+k+\epsilon}}\right),
\label{psik}
\end{equation}
along $\{ r=R\}$ hypersurfaces  and 
\begin{equation}
r{\psi_k}|_{\mathcal{I}^{+}}= c_k\cdot  I_{0}^{(k+1)}[{\psi_k}]\cdot \frac{1}{\tau^{2+k}}+O\left(\frac{1}{\tau^{2+k+\epsilon}}\right) 
\label{rpsik}
\end{equation}
 along the future null infinity $\mathcal{I}^{+}$, 
for some numerical constants $C_k,c_k$ that were explicitly computed in \cite{paper2}. 
 Note that we have  
\begin{equation}
\psi_k=\psi-I_{0}^{(1)}[\psi_0] \cdot T\boldsymbol{\psi}_{\text{seed}}-I_{0}^{(2)}[\psi_1] \cdot T^2\boldsymbol{\psi}_{\text{seed}}-\cdots-I_{0}^{(k)}[\psi_{k-1}] \cdot T^k\boldsymbol{\psi}_{\text{seed}}.
\label{ak}
\end{equation}
In view of \eqref{psik}, \eqref{rpsik} and \eqref{ak} we obtain:
\begin{equation}
\begin{split}
\psi=& I_{0}^{(1)}[\psi_0] \cdot T\boldsymbol{\psi}_{\text{seed}}+I_{0}^{(2)}[\psi_1] \cdot T^2\boldsymbol{\psi}_{\text{seed}}+\cdots+I_{0}^{(k)}[\psi_{k-1}] \cdot T^k\boldsymbol{\psi}_{\text{seed}}\\ &
+ C_k\cdot  I_{0}^{(k+1)}[\psi_k]\cdot \frac{1}{\tau^{3+k}}+O\left(\frac{1}{\tau^{3+k+\epsilon}}\right)
\end{split}
\label{psifinalk}
\end{equation}
asymptotically in time along $\{r=R\}$ hypersurfaces and 
\begin{equation}
\begin{split}
r\psi=& I_{0}^{(1)}[\psi_0] \cdot T\boldsymbol{r\psi}_{\text{seed}}+I_{0}^{(2)}[\psi_1] \cdot T^2\boldsymbol{r\psi}_{\text{seed}}+\cdots+I_{0}^{(k)}[\psi_{k-1}] \cdot T^k\boldsymbol{r\psi}_{\text{seed}}\\&
+ c_k\cdot  I_{0}^{(k+1)}[\psi_k]\cdot \frac{1}{\tau^{2+k}}+O\left(\frac{1}{\tau^{2+k+\epsilon}}\right)
\end{split}
\label{rpsifinalk}
\end{equation}
asymptotically in time along null infinity.

Clearly, the constants $I_{0}^{(i)}[\psi_{i-1}], i=1,2,\cdots, k$ depend on $\psi$. On the other hand, the function $\boldsymbol{\psi}_{\text{seed}}$, its radiation field and their time derivatives are completely independent of $\psi$. Hence, by virtue of \eqref{psifinalk} and \eqref{rpsifinalk}, \textit{the asymptotic expansions of the scalar field $\psi$ and its radiation field $r\psi$ involve only the constants $I_{0}^{(i)}[\psi_{i-1}], i=1,2,\cdots, k$, which can be explicitly determined by the initial data of $\psi$ (and the choice of $\boldsymbol{\psi}_{\text{seed}}$). The exact asymptotic expansions of $\psi$ and its radiation field are obtained from the asymptotic expansions of the \textbf{fixed} function $\boldsymbol{\psi}_{\text{seed}}$, its radiation field and their time derivatives. } We remark that in view of Theorem \ref{thm1}, the first two terms in the late-time asymptotic expansion of $\boldsymbol{\psi}_{\text{seed}}$ along $\mathcal{I}^+$ can be determined explicitly.

\section{Preliminaries}
\label{preliminaries}

\subsection{The geometric assumptions}
\label{thegeometricassumptions}

In this paper we consider the spacetimes $(\mathcal{M},g)$ of Section 2.1 of \cite{paper2}
. These spacetimes are stationary, spherically symmetric and asymptotically flat and, in particular, include the Schwarzschild family and the larger sub-extremal Reissner--Nordstr\"{o}m family of black holes as special cases. In this section we briefly recall the geometric assumptions on the spacetime metrics and introduce the notation that we use in this paper. 

\subsubsection{The Lorentzian manifolds $(\mathcal{M},g)$}
\label{thegeometricassumptions1}

The manifold $\mathcal{M}$ is covered by appropriate double null coordinates $(u,v,\theta,\varphi)$ with respect to which the metric takes the form 
\begin{equation*}
g=-D(r)dudv+r^2(d\theta^2+\sin^2 \theta d\varphi^2),
\end{equation*}
with $D(r)$ a smooth function such that
\begin{equation}
D(r)=1-\frac{2M}{r}+\frac{d_1}{r^2}+O_3(r^{-2-\beta})
\label{dbehavior}
\end{equation}
where $d_1\in\mathbb{R}$. Here $r=r(u,v)$ denotes the area-radius of the spheres $S_{u,v}$ of symmetry. 
The null hypersurfaces $C_{u_{0}}=\{u=u_0\}$ terminate in the future (as $r,v\rightarrow +\infty$) at future null infinity  $\mathcal{I}^{+}$. Note also that $u$ is a ``time'' parameter along the future null infinity $\mathcal{I}^{+}$ such that $u$ increases towards the future. 

Furthermore, by an appropriate normalization (see, for instance, \cite{paper2}) we can assume that 
\begin{equation}
v-u=2r^*
\label{uvr}
\end{equation} 
where the function $r^*=r^{*}(r)$ is given by
\[ r^*=R+\int_{R}^rD^{-1}(r')dr'.\]
Here $R>0$ is a sufficiently large but fixed constant. Since the bulk of the analysis of the present paper occurs in the region $\{r\geq R\}$ we start by proving the following lemma for this region:
\begin{lemma}
\label{lm:estr}
There exists a constant $C_R>0$ such that we can estimate in $\{r\geq R\}\cap\{u\geq 0\}$:
\begin{align*}
\left|r-\frac{v-u}{2}+2M\log (v-u)\right|\leq&\: C_R,\\
\left|r^{-1}-\frac{2}{v-u}-8M(v-u)^{-2}\log (v-u)\right|\leq&\: C_R(v-u)^{-2}\\
\left|r^{-2}-\frac{4}{(v-u)^2}-32M(v-u)^{-3}\log (v-u)\right|\leq&\: C_R (v-u)^{-3},\\
r^{-3}(u,v)\leq &\:C_R (v-u)^{-3}.
\end{align*}
Let $\gamma=\{u=\frac{v}{2}+R_*\}$, where $R_*=r_*(R)$. There exists a constant $c_R>0$ such that
\begin{equation*}
\begin{array}{ll}
c_R \cdot v \leq v-u-1 \leq v\quad &\textnormal{if}\quad  0\leq u\leq u_{\gamma}(v):=\frac{v}{2}+R_*,\\
c_R\cdot v\leq u\leq v \quad  &\textnormal{if}\quad  u_{\gamma}(v)\leq u\leq v-2R_*.
\end{array}
\end{equation*}
\end{lemma}
\begin{proof}
Let $r_0\in (r_{\min},\infty)$ be arbitrary. We have that
\begin{equation*}
\frac{v-u}{2}=r_*(r)=\int_{r_0}^r D^{-1}(r')\,dr'=\int_{r_0}^r 1+2M/r'+O(r'^{-1-\beta})\,dr'=r+2M\log r+O(1).
\end{equation*}
Hence, in $\{r\geq R\}$, we can conclude that there exist constants $c_R,C_R>0$ such that
\begin{equation*}
c_R\cdot (v-u)\leq r\leq C_R (v-u).
\end{equation*}
Moreover, 
\begin{equation*}
\left|r-\frac{v-u}{2}+2M\log (v-u)\right|\leq C_R.
\end{equation*}
Hence
\begin{equation*}
\begin{split}
\frac{1}{r}=&\:\frac{2}{v-u}\left(\frac{1}{1+4M(v-u)^{-1}\log(v-u)+O((v-u)^{-1})}\right)\\
=&\:\frac{2}{v-u}\left(1-\frac{4M(v-u)^{-1}\log(v-u)+O((v-u)^{-1})}{1+4M(v-u)^{-1}\log(v-u)+O((v-u)^{-1})}\right).
\end{split}
\end{equation*}
This implies that there exists a $C>0$ such that
\begin{equation*}
\left|r^{-1}-\frac{2}{v-u}-8M(v-u)^{-2}\log (v-u)\right|\leq C(v-u)^{-2}.
\end{equation*}
and consequently,
\begin{equation*}
\left|r^{-2}-\frac{4}{(v-u)^2}-32M(v-u)^{-3}\log (v-u)\right|\leq C(v-u)^{-3}.
\end{equation*}
\end{proof}

For the region $\{r\leq R\}$ we consider two cases depending on the value of\\ $r_{\text{min}}=\inf_{\mathcal{M}}r<R$. 

Case I: $r_{\text{min}}>0$ (see Figure \ref{fig:fullfoliations22} below). In this case, we assume that $D(r_{\text{min}})=0$ and that $\left.\frac{dD(r)}{dr}\right|_{r=r_{\text{min}}}\neq 0$ and $D(r)>0$ for $r>r_{\text{min}}$. The boundary hypersurface $\{r=r_{\text{min}}\}\cap\{v<\infty \}$ is known as the future event horizon and is denoted by $\mathcal{H}^{+}$.  The sub-extremal Reissner-Nordstr\"{o}m is a special case with $D(r)=1-\frac{2M}{r}+\frac{e^2}{r^2}$ with $|e|<M$.

Case II: $r_{\text{min}}=0$ (see Figure \ref{fig:fullfoliations12} below). In this case, we assume that $\inf_{\mathcal{M}}D(r)>0$. The Minkowski spacetime is a special case with $D(r)=1$. 

\subsubsection{The spacelike-null foliation $\Sigma_{\tau}$}
\label{sec:ASpacelikeNullFoliation}

 The stationary Killing field $T$ is equal to 
\[T=\partial_v+\partial_u.\] 

We define the outgoing null hypersurfaces $$\mathcal{N}_{u'}=\{(u,r,\theta,\varphi)\,:\, u=u',\: r\geq R\}$$ and the regions $$\mathcal{A}=\bigcup_{u\in[0,\infty)}\mathcal{N}_{u}$$ and $$\mathcal{A}^{u_2}_{u_1}=\bigcup_{u\in [u_1,u_2]}\mathcal{N}_{u}.$$ 

 We  define the spacelike-null hypersurface $\Sigma_{0}$ as follows:
\begin{equation*}
\Sigma_0=\{(v,r,\theta,\varphi)\,;\, v=v_{\Sigma_{0}}(r), r\leq R\}\cup \mathcal{N}_0,
\end{equation*}
where $v_{\Sigma_{0}}(r)$ is the smooth function satisfying:
\begin{align*}
\frac{dv_{\Sigma_0}}{dr}=&\:h_{\Sigma_0},\\
v_{\Sigma_0}(R)=&\:2r_*(R),
\end{align*}
where  $h_{\Sigma_{0}}: [r_{\text{min}},R] \rightarrow \mathbb{R}$ be a smooth, positive function.
Note that $(v,v_{\Sigma_0}(R),\theta,\varphi)\in \mathcal{N}_0$ and hence $\Sigma_{0}$ is spacelike for $r\leq R$ and null for $r\geq R$.

Consider the future $J^+(\Sigma_0)$ of $\Sigma_{0}$. We assume that $\inf_{J^+(\Sigma_0)}v=v_0>0$. Hence,\textit{ in the sub-extremal Reissner--Nordstr\"{o}m  black hole case, the hypersurface $\Sigma_0$ intersects the future event horizon to the future of the bifurcation sphere}. Moreover, we define the smooth function $\tau $ on $J^+(\Sigma_0)$, such that $\tau|_{{\Sigma_0}}=0$, and $T(\tau)=1$. In $\mathcal{A}$ we clearly have $\tau=u$.

We finally define $\Sigma_{\tau}$ to be the level sets of the function $\tau$ and denote
\begin{equation*}
\mathcal{R}=J^+(\Sigma_0)=\bigcup_{\tau\in [0,\infty)}\Sigma_{\tau}.
\end{equation*}
By stationarity, we have that $\Sigma_{\tau}$ is isometric to $\Sigma_0$ for all $\tau\geq 0$. 
\begin{figure}[H]
\centering
\includegraphics[width=3in]{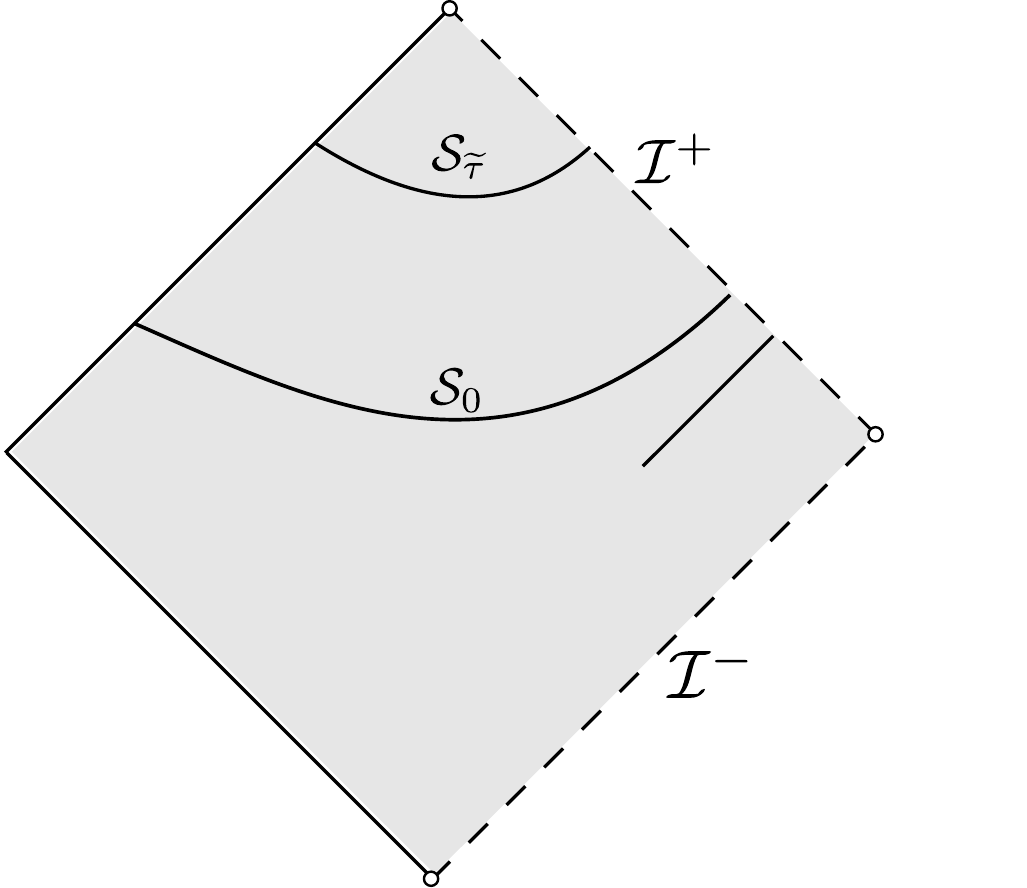}
\caption{\label{fig:fullfoliations22}Penrose diagram of $\mathcal{M}$ in the case $r_{\text{min}}>0$. The hypersurfaces $\Sigma_{\tau}$ intersect the future event horizon $\mathcal{H}^{+}$.}
\end{figure}
\begin{figure}[H]
\centering
\includegraphics[width=3in]{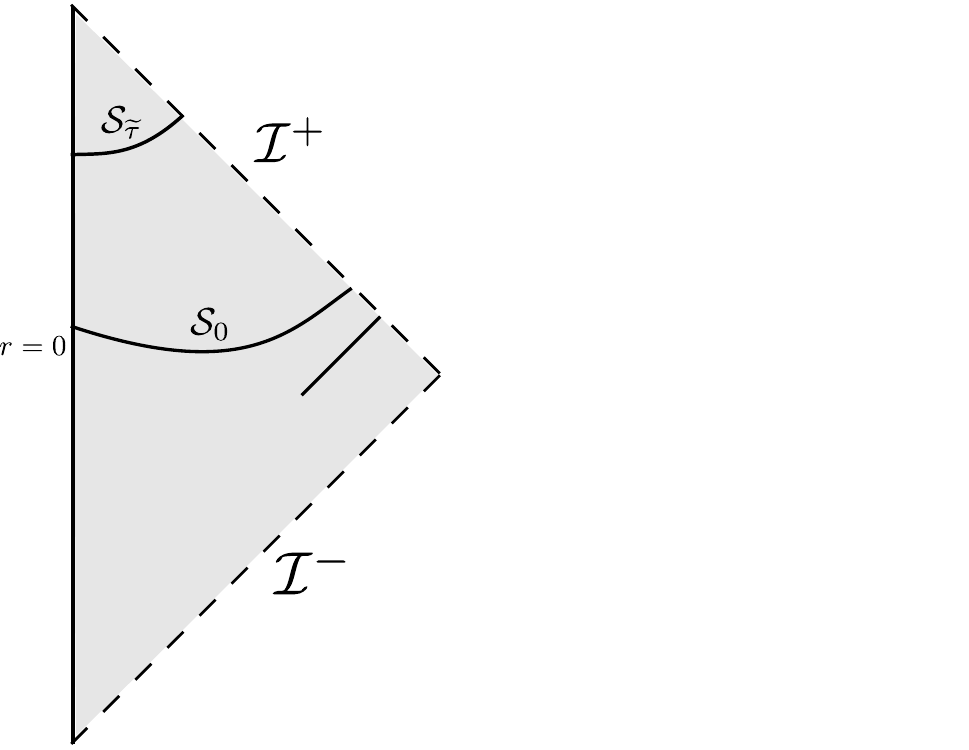}
\caption{\label{fig:fullfoliations12} The Penrose diagram of $\mathcal{M}$ in the case $r_{\text{min}}=0$.}
\end{figure}
We can also consider the coordinate chart $(\tau,\rho,\theta,\varphi)$ in $\mathcal{R}\cap\{r\leq R\}$, where $\rho=r|_{\Sigma_0}$. Then we can express:

\begin{align*}
\partial_{\tau}=&\:T,\\
\partial_{\rho}=&-2D^{-1}\partial_{u}+h_{\Sigma_0}T.
\end{align*}

\subsection{Assumptions for the wave equation}
\label{ivp}
We consider smooth initial data on $\Sigma_{0}$ for the wave equation. We consider that the background geometry is such that the following estimates hold for solutions to the wave equation:

\begin{enumerate}
	\item Uniform boundedness of the energy fluxes
	\item Local integrated energy decay in the form of local and global Morawetz estimates. 
\end{enumerate}
The precise estimates that we need to assume are contained in Section 2.4 in \cite{paper2}.

\subsection{The Newman--Penrose constant $I_{0}[\psi]$ at null infinity}
\label{np}

Let $\psi$ be a smooth spherically symmetric solution to the wave equation \eqref{we}
such that there is a finite constant $I_0$ such that
\begin{equation}
\label{eq:rasyminitdata}
\partial_r(r\psi)|_{u=0}(u=0,r)=I_0\cdot r^{-2}+O(r^{-3}).
\end{equation}
Note that \[\partial_r=\frac{2}{D}\partial_v.\]
The constant $I_0=I_0[\psi]$ is known as the Newman--Penrose constant and has the following conservation property:
\[\lim_{r\rightarrow \infty}r^2\partial_r(r\psi)|_{u=u_0}(r)=I_0[\psi]\]
for all $u_0$.  Furthermore, it always holds that \[I_0[T\psi]=0.\]

It will be convenient to relate the $r$-asymptotics of the initial data in \eqref{eq:rasyminitdata} to asymptotics in $v$.
\begin{lemma}
Let $\psi$ satisfy \eqref{eq:rasyminitdata} at $u=0$. Then there exists a constant $C_R>0$, such that we can estimate along $\mathcal{N}_0$:
\begin{equation}
\label{eq:vasyminitdata}
\left|\partial_v(r\psi)(0,v)-2I_0v^{-2}-16MI_0v^{-3}\log v\right|\leq C_R\cdot (P[\psi]+I_0[\psi]),
\end{equation}
where the constant $P[\psi]>0$ is defined in \eqref{np1-der1}.
\end{lemma}
\begin{proof}
By the definition of $P[\psi]$ in  \eqref{np1-der1} and by \eqref{dbehavior}, we can estimate along $\mathcal{N}_0$
\begin{equation*}
\left|\partial_r(r\psi)(0,v)-I_0r^{-2}\right|\leq C_R\cdot P[\psi]r^{-3},
\end{equation*}
so, using Lemma \ref{lm:estr} to estimate $r^{-3}\leq C v^{-3}$ along $\{u=0\}$, together with the fact that $\partial_r=2D^{-1}\partial_v=(2+O(r^{-1})\partial_v$ in $(u,r)$ coordinates, we obtain
\begin{equation*}
\left|\partial_v(r\psi)(0,v)-\frac{1}{2}I_0r^{-2}\right|\leq C_R(I_0+P[\psi])v^{-3}.
\end{equation*}
We use Lemma \ref{lm:estr} once more to estimate
\begin{equation*}
\left|\frac{1}{2}I_0r^{-2}-2I_0v^{-2}-16Mv^{-3}\log v\right|\leq C_RI_0v^{-3}.
\end{equation*}
The estimate \eqref{eq:vasyminitdata} then follows from the triangle inequality.
\end{proof}
\subsection{The time-integral $\psi^{(1)}$ and the time-inverted Newman--Penrose constant $I_{0}^{(1)}$  }
\label{timeinverted}
Let $\psi$ be a smooth spherically symmetric solution to the wave equation \eqref{we} with vanishing Newman--Penrose constant $I_0 [ \psi ] = 0$. In fact, we assume that along $\Sigma_{0}$ we have 
\[\lim_{r\rightarrow \infty}r^3\partial_{r}(r\psi)|_{\Sigma_0} <\infty.\]
The time-integral $\psi^{(1)}$ is defined to be the unique regular solution to the wave equation \eqref{we} such that 
\[\lim_{r\rightarrow \infty}\left.\psi^{(1)}\right|_{\Sigma_{0}}=0 \]
and
\begin{equation}\label{timeintegral}
T \psi^{(1)} = \psi,
\end{equation}
in $\mathcal{R}$. See Proposition 9.1 of \cite{paper2}. 

The time-inverted Newman--Penrose constant $I_{0}^{(1)}[\psi]$ is defined to be the Newman--Penrose constant of the time-integral $\psi^{(1)}$, that is
\[I_{0}^{(1)}[\psi]=I_{0}[\psi^{(1)}].  \]
A calculation (see \cite{paper2}) yields
\begin{equation}
\begin{split}
I_0[\psi^{(1)}]=& -2\lim_{r\to \infty}r^3\partial_{v}\phi|_{\Sigma_{0}}+M R(2-Dh_{\Sigma_{0}}(R))\phi(0,R)+2M\int_{r\geq R}r \partial_{v}\phi\Big|_{\mathcal{N}_{0}}\,dv'\\
&-M\int_{r_{\rm min}}^R 2(1-h_{\Sigma_{0}}D)r\partial_{\rho}\phi-(2-Dh_{\Sigma_{0}})rh_{\Sigma_{0}} T\phi-(r\cdot (Dh_{\Sigma_{0}})')\cdot\phi\Big|_{\Sigma_{0}}\,d\rho',
\end{split}
\label{nptif1}
\end{equation}
where $\phi=r\psi$. In particular, if the initial data for $\psi$ is compactly supported in $\{r_{\text{min}}\leq r<R\}$, we have that
\begin{equation}
\begin{split}
I_0[\psi^{(1)}]=&\:-M\int_{r_{\rm min}}^R 2(1-h_{\Sigma_{0}}D)r\partial_{\rho}\phi-(2-Dh_{\Sigma_{0}})rh_{\Sigma_{0}} T\phi-(r\cdot (Dh_{\Sigma_{0}})')\cdot\phi\Big|_{\Sigma_{0}}\,d\rho'.
\end{split}
\label{npfcompact}
\end{equation}

\subsection{Leading-order asymptotics}
\label{asymptotics}

From \cite{paper2} we have the following leading-order asymptotics for solutions $\psi$ to the wave equation  and their $T$-derivative $T\psi$ in $\{r\geq R\}\cap\{u\geq 0\}$. 
\begin{theorem}
\label{thm:loasymp}
Consider initial data for $\psi$ on the hypesurface $\Sigma_0$ such that $I_0[\psi]<\infty$, ${E}[\psi]<\infty$ and $P[\psi]<\infty$ and for which the assumptions of Section \ref{ivp} hold. Then
there exists $\epsilon>0$ suitably small such that if $\beta>\epsilon$ we have that for all $(u,v)\in \{r\geq R\}\cap\{u\geq 0\}$ we have:
\begin{equation}
\label{eq:loasymppsi}
\left|\psi(u,v)-\frac{4I_0[\psi]}{(u+1)v}\right|\leq C_R(u+1)^{-1-\epsilon}v^{-1}\left[I_0[\psi]+\sqrt{E}[\psi]+P[\psi]\right].\\
\end{equation}
If moreover ${E}_T[\psi]<\infty$ and $P_{T}[\psi]<\infty$, then there exists $\epsilon>0$ suitably small such that if $\beta>\epsilon$ we have that for all $(u,v)\in \{r\geq R\}\cap\{u\geq 0\}$:
\begin{equation}
\label{eq:loasympTpsi}
\left|T\psi(u,v)+\frac{4I_0[\psi]}{(u+1)^2v}\left(1+\frac{u+1}{v}\right)\right|\leq C_R(u+1)^{-2-\epsilon}v^{-1}\left[I_0[\psi]+\sqrt{E_T}[\psi]+P_{T}[\psi]\right].
\end{equation}
\end{theorem}
Recall that the norms $E,E_T,P,P_T$ are defined in \ref{energynorms}.

\section{Higher-order estimates in the region $\{ r \geq R \}$}\label{ho}
This section concerns estimates about \textit{spherically symmetric} solutions. 
\subsection{Estimates for $\partial_v(r\psi)$ in the case $I_0 [\psi ]  \neq 0$}
\begin{proposition}
\label{prop:asymdvphi}
For all spherically symmetric solutions $\psi$ to the wave equation \eqref{we} on the $(\mathcal{M},g)$ backgrounds of Section \ref{preliminaries} with non-vanishing Newman--Penrose constant $I_0$ there exists a constant $C_R>0$ such that we can estimate in $\{r\geq R\}\cap\{u\geq 0\}$:
\begin{equation}
\label{eq:2ndoasympdvphi}
\begin{split}
&\Bigg|\partial_v(r\psi)(u,v)-2I_0v^{-2}-16M I_0v^{-3}\log v+8MI_0(u+1)v^{-3}(v-u-1)^{-1}\\
&+8MI_0v^{-3} \log \left(\frac{v(u+1)}{v-(u+1)}\right)\Bigg|\\
\leq&\: C_R(I_0+\sqrt{E}+P)\textnormal{Err}_{\beta}(u,v),
\end{split}
\end{equation}
where
\begin{equation*}
\textnormal{Err}_{\beta}(u,v):=v^{-3}+v^{-2-\epsilon}\cdot (v-u-1)^{-1}+v^{-2}\cdot (v-u-1)^{-2+\max\{1-\beta,\eta\}},
\end{equation*}
with $\eta>0$ arbitrarily small.
\end{proposition}
\begin{proof}
We use \eqref{we} and \eqref{dbehavior} together with the estimates in Lemma \ref{lm:estr} to obtain:
\begin{equation*}
\begin{split}
\partial_u\partial_v(r\psi)(u,v)=&-\frac{1}{4}DD'\cdot \psi\\
=&-\frac{1}{4}\left(1-\frac{2M}{r}+O_3(r^{-1-\beta})\right)\cdot\left(\frac{2M}{r^2}+O_2(r^{-2-\beta})\right)\cdot \psi\\
=&\:\left(-\frac{1}{2}Mr^{-2}+O_2(r^{-2-\beta})\right)\cdot \psi\\
=&\:\Big[-2M(v-u-1)^{-2}+O((v-u-1)^{-2-\beta})\\
&+\log(v-u-1)O((v-u-1)^{-3})\Big]\cdot \psi\\
=&\:\Big[-2M(v-u-1)^{-2}+O((v-u-1)^{-3+\max\{1-\beta,\eta\}})\Big]\cdot \psi,
\end{split}
\end{equation*}
with $\eta>0$ arbitrarily small.

By \eqref{eq:loasymppsi} we can therefore estimate
\begin{equation*}
\begin{split}
\left|\partial_u\partial_v(r\psi)(u,v)+\frac{8MI_0}{v(u+1)}(v-u-1)^{-2}\right|\leq &\:C_RI_0(v-u-1)^{-3+\max\{1-\beta,\eta\}}v^{-1}(u+1)^{-1}\\
&+ C_R(I_0+\sqrt{E}+P)(v-u-1)^{-2}v^{-1}(u+1)^{-1-\epsilon}.
\end{split}
\end{equation*}

By applying the fundamental theorem of calculus in $u$, we have that
\begin{equation*}
\begin{split}
\Big|\partial_v(r\psi)(u,v)&-\partial_v(r\psi)(0,v)+8MI_0v^{-1}\int_0^u (u'+1)^{-1}(v-u'-1)^{-2}\,du'\Big|\\
\leq&\:  C_R(I_0+\sqrt{E}+P) (J_1(u,v)+J_2(u,v)),
\end{split}
\end{equation*}
where
\begin{align*}
J_1(u,v):=&\:v^{-1}\int_0^u (v-(u'+1))^{-3+\max\{1-\beta,\eta\}} (u'+1)^{-1}\,du',\\
J_2(u,v):=&\:v^{-1}\int_0^u (v-(u'+1))^{-2} (u'+1)^{-1-\epsilon}\,du'.
\end{align*}
In order to estimate the integrals $J_1(u,v)$ and $J_2(u,v)$ we consider the curve $\gamma=\{ v'=2R^*+2u' \}$ depicted below in Figure \ref{fig:rgeqR}:
\begin{figure}[H]
\begin{center}
\includegraphics[width=3in]{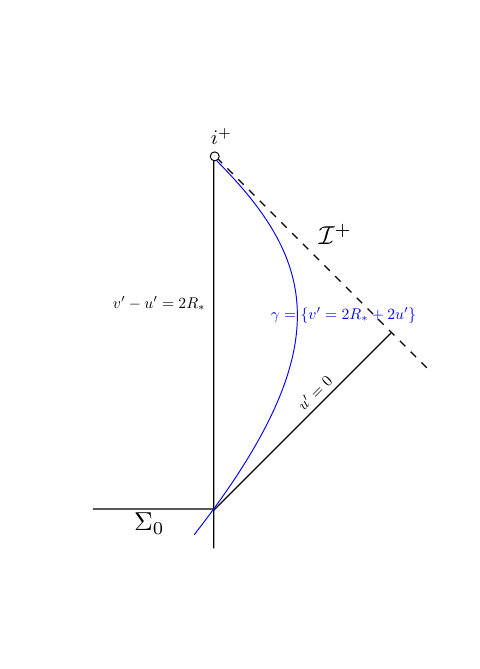}
\caption{\label{fig:rgeqR} The region $\{r\geq R\}\cap\{u\geq 0\}$ and the curve $\gamma$.}
\end{center}
\end{figure}
\underline{\textbf{Step 1:} Estimating $J_1(u,v)$ and $J_2(u,v)$}\\
\\
We estimate $J_1(u,v)$ and $J_2(u,v)$ by partitioning the integration range as $[0,u]=[0,u_{\gamma}(v)]\cup (u_{\gamma}(v),u]$, where $u_{\gamma}(v)=\frac{v}{2}-R_*$, and applying the estimates in Lemma \ref{lm:estr} corresponding to the intervals of the partition. 

We obtain:
\begin{equation*}
\begin{split}
J_1(u,v)=&\: v^{-1}\int_0^u (v-(u'+1))^{-3+\max\{1-\beta,\eta\}}(u'+1)^{-1}\,du'\\
=&\:v^{-1}\int_0^{\frac{v}{2}-R_*} (v-(u'+1))^{-3+\max\{1-\beta,\eta\}}(u'+1)^{-1}\,du'\\
&+v^{-1}\int_{\frac{v}{2}-R_*}^u (v-(u'+1))^{-3+\max\{1-\beta,\eta\}}(u'+1)^{-1}\,du'\\
\leq&\: C_Rv^{-4+\max\{1-\beta,\eta\}}\log v+ C_Rv^{-2}\cdot (v-u-1)^{-2+\max\{1-\beta,\eta\}},
\end{split}
\end{equation*}
and similarly,
\begin{equation*}
\begin{split}
J_2(u,v)=&\: v^{-1}\int_0^u (v-(u'+1))^{-2}(u'+1)^{-1-\epsilon}\,du'\\
=&\:v^{-1}\int_0^{\frac{v}{2}-R_*} (v-(u'+1))^{-2}(u'+1)^{-1-\epsilon}\,du'\\
&+v^{-1}\int_{\frac{v}{2}-R_*}^u (v-(u'+1))^{-2}(u'+1)^{-1-\epsilon}\,du'\\
\leq&\: C_Rv^{-3}+ C_Rv^{-2-\epsilon}\cdot (v-u-1)^{-1}.
\end{split}
\end{equation*}
\\
\underline{\textbf{Step 2:} Evaluating $\int_0^u (u'+1)^{-1}(v-u'-1)^{-2}\,du'$}\\
\\
We can write:
\begin{equation*}
8MI_0 v^{-1}\int_0^u (u'+1)^{-1}(v-(u'+1))^{-2}\,du'=8MI_0 v^{-1}\int_1^{u+1} x^{-1}(v-x)^{-2}\,dx,
\end{equation*}
where $x=u'+1$.

Let
\begin{equation*}
F(x;v)=v^{-1}\cdot (v-x)^{-1}-v^{-2}\cdot \log\left(\frac{v-x}{x}\right),
\end{equation*}
then
\begin{equation*}
\begin{split}
F'(x;v)=&\:v^{-1}\cdot (v-x)^{-2}-v^{-2}\cdot\frac{x}{v-x} \left(-vx^{-2}\right)\\
=&\:x^{-1}(v-x)^{-2}.
\end{split}
\end{equation*}
And hence,
\begin{equation*}
\begin{split}
8MI_0& v^{-1}\int_0^u (u'+1)^{-1}(v-(u'+1))^{-2}\,du'\\
=&\:8MI_0v^{-1} F(x;v)\big|^{x=u+1}_{x=1}\\
=&\:8MI_0v^{-2} \left[((v-(u+1))^{-1}-v^{-1}\cdot  \log\left(\frac{v-(u+1)}{u+1}\right)-(v-1)^{-1}+v^{-1}\cdot  \log\left(v-1\right)\right]\\
=&\:8MI_0v^{-2}\left[\frac{u+1}{v(v-(u+1))}-v^{-1}\cdot  \log\left(\frac{v-(u+1)}{v(u+1)}\right)\right]+O(v^{-4}),
\end{split}
\end{equation*}
where we used that
\begin{equation*}
(v-(u+1))^{-1}-v^{-1}=\frac{u+1}{v(v-(u+1))}
\end{equation*}
and
\begin{equation*}
v^{-2}\left(v^{-1}-(v-1)^{-1}+v^{-1}\log\left(\frac{v-1}{v}\right)\right)=O(v^{-4}).
\end{equation*}
By combining the above expression with the estimates for $J_1$ and $J_2$, we arrive at:
\begin{equation*}
\begin{split}
\Bigg|\partial_v(r\psi)(u,v)&-\partial_v(r\psi)(0,v)+8MI_0(u+1)v^{-3}(v-u-1)^{-1}+8MI_0v^{-3} \log \left(\frac{v(u+1)}{v-(u+1)}\right)\Bigg|\\
\leq&\:  C_R(I_0+\sqrt{E}+P)[v^{-3}+v^{-2-\epsilon}\cdot (v-u-1)^{-1}+v^{-2}\cdot (v-u-1)^{-2+\max\{1-\beta,\eta\}}].
\end{split}
\end{equation*}
Finally, we obtain \eqref{eq:2ndoasympdvphi} by applying \eqref{eq:vasyminitdata}:
\begin{equation*}
\left|\partial_v(r\psi)(0,v)-2I_0v^{-2}-16MI_0v^{-3}\log v\right|\leq C_R\cdot (P[\psi]+I_0).
\end{equation*}
\end{proof}

\subsection{Estimate for $r\psi$ in the case $I_0 [\psi ] \neq 0$}
We now integrate the equation \eqref{eq:2ndoasympdvphi} starting from the curve $\{r=R\}$ to the hypersurface $\{v'=v\}$.
\begin{proposition}
\label{prop:asymphi}For all spherically symmetric solutions $\psi$ to the wave equation \eqref{we} on the $(\mathcal{M},g)$ backgrounds of Section \ref{preliminaries} with non-vanishing Newman--Penrose constant $I_0$  there exists a constant $C_R>0$ such that we can estimate in $\{r\geq R\}\cap\{u\geq 0\}$:
\begin{equation}
\label{eq:asymphi}
\begin{split}
\Bigg|r\psi(u,v)&-2I_0[(u+1)^{-1}-v^{-1}]+4MI_0(u+1)^{-2}\log(u+1)-4MI_0v^{-2}\log(u+1)\\
&+8MI_0v^{-2}\log v+4MI_0[(u+1)^{-2}+v^{-2}]\log\left(\frac{v-u-1}{v}\right)\Bigg|\\
\leq&\:  C_R(I_0+\sqrt{E}+P)(u+1)^{-2}.
\end{split}
\end{equation}
\end{proposition}

\begin{proof}
We apply the fundamental theorem of calculus in the $v$-direction to obtain:
\begin{equation*}
r\psi(u,v)=r\psi(u,u+2R_*)+\int_{u+2R_*}^{v}\partial_v(r\psi)(u,v')\,dv'.
\end{equation*}
Furthermore, by \eqref{eq:2ndoasympdvphi} we can estimate
\begin{equation*}
\begin{split}
&\left|\int_{u+2R_*}^{v}\partial_v(r\psi)(u,v')\,dv'-\int_{u+2R_*}^{v} 2I_0v'^{-2}+16MI_0v'^{-3}\log v'\,dv'+J_3(u,v)+J_4(u,v)\right|\\
\lesssim&\: \int_{u+2R_*}^{v} \textnormal{Err}_{\beta}(u,v')\,dv',
\end{split}
\end{equation*}
where
\begin{align*}
J_3(u,v)=&\:8MI_0(u+1)\int_{u+2R_*}^{v}v'^{-3}(v'-u-1)^{-1}\,dv',\\
J_4(u,v)=&\:8MI_0\int_{u+2R_*}^{v}v'^{-3} \log \left(\frac{v'(u+1)}{v'-(u+1)}\right)\,dv'.
\end{align*}
First of all, we have that
\begin{equation*}
\left|-\int_{u+2R_*}^{v} 2I_0v'^{-2}\,dv'+2I_0[(u+1)^{-1}-v^{-1}]\right|\leq CI_0(u+1)^{-2}.
\end{equation*}
Let
\begin{equation*}
F_1(v):=-\frac{1}{2}v^{-2}\log v-\frac{1}{4}v^{-2}.
\end{equation*}
Then
\begin{equation*}
F_1'(v):=v^{-3}\log v
\end{equation*}
and hence
\begin{equation*}
\begin{split}
-\int_{u+2R_*}^{v} 16MI_0v'^{-3}\log v'\,dv'=&-16MI_0F_1(v)+16MI_0F_1(u+2R_*)\\
=&\:-8MI_0[(u+2R_*)^{-2}\log(u+2R_*)-v^{-2}\log v].
\end{split}
\end{equation*}
and we can conclude that
\begin{equation*}
\begin{split}
\left|-\int_{u+2R_*}^{v} 16MI_0v'^{-3}\log v'\,dv'+8MI_0[(u+1)^{-2}\log(u+1)-v^{-2}\log v]\right|\leq CI_0(u+1)^{-2}.
\end{split}
\end{equation*}
\\\
\\
\underline{\textbf{Step 1:} Estimating the integral of  $\textnormal{Err}_{\beta}(u,v)$}\\
\\
We immediately obtain:
\begin{equation*}
\int_{u+2R_*}^{v}v'^{-3}\,dv'\leq C_R ((u+1)^{-2}-v^{-2}).
\end{equation*}

We estimate the remaining terms in the integral of $\textnormal{Err}_{\beta}(u,v)$ by partitioning the $v$-integration range as $[u+2R_*,\infty)=[u+2R_*,2u+2R_*]\cup (2u+2R_*,\infty)$, where we assume without loss of generality that $v\geq 2u+2rR_*$, and making use of the estimates in Lemma \ref{lm:estr}. 

We first estimate:
\begin{equation*}
\begin{split}
\int_{u+2R_*}^v v'^{-2-\epsilon}\cdot (v'-u-1)^{-1}\,dv'=&\:\int_{u+2R_*}^{2u+2R_*} v'^{-2-\epsilon}\cdot (v'-u-1)^{-1}\,dv'\\
&+\int_{2u+2R_*}^{v} v'^{-2-\epsilon}\cdot (v'-u-1)^{-1}\,dv'\\
\lesssim&\: (u+1)^{-2-\epsilon}\cdot \log (u+1)+(u+1)^{-2-\epsilon}\\
\lesssim&\: (u+1)^{-2-\epsilon}\cdot \log (u+1).
\end{split}
\end{equation*}

Furthermore,
\begin{equation*}
\begin{split}
\int_{u+2R_*}^{v} v'^{-2}\cdot (v'-u-1)^{-2+\max\{1-\beta,\eta\}}\,dv'=&\:\int_{u+2R_*}^{2u+2R_*} v'^{-2}\cdot (v'-u-1)^{-2+\max\{1-\beta,\eta\}}\,dv'\\
&+\int_{2u+2R_*}^{v} v'^{-2}\cdot (v'-u-1)^{-2+\max\{1-\beta,\eta\}}\,dv'\\
\lesssim&\: (u+1)^{-2}+(u+1)^{-3+\max\{1-\beta,\eta\}}\\
\lesssim&\:  (u+1)^{-2},
\end{split}
\end{equation*}
for $\eta>0$ suitably small.
 \\
 \\
\underline{\textbf{Step 2:} Evaluating $J_3$}\\
\\
Let
\begin{equation*}
F_2(v;c)=c^{-3}\log(1-cv^{-1})+c^{-2}v^{-1}+\frac{1}{2}v^{-2}c^{-1}.
\end{equation*}
Then
\begin{equation*}
\begin{split}
\frac{dF_2}{dv}(v;c)=&\:c^{-3}\frac{v}{v-c}\cdot cv^{-2}-c^{-2}v^{-2}-v^{-3}c^{-1}\\
=&\:\frac{v^2-v(v-c)-c(v-c)}{c^2(v-c)v^3}\\
=&\: v^{-3}(v-c)^{-1}.
\end{split}
\end{equation*}
Hence,
\begin{equation*}
\begin{split}
J_3(u,v)=&\:8MI_0(u+1)\int_{u+2R_*}^{v}v'^{-3}(v'-u-1)^{-1}\,dv'=8MI_0(u+1)F_2(v';u+1)\big|^{v'=v}_{v'=u+2R_*}\\
=&\:8MI_0\left[(u+1)^{-2} \log\left(\frac{u+2R_*}{2R_*-1}\right)-(u+1)^{-1}(u+2R_*)^{-1}-\frac{1}{2}(u+2R_*)^{-2}(u+1)^{-1}\right]\\
&+8MI_0\left[(u+1)^{-2}\log\left(\frac{v-u-1}{v}\right)+(u+1)^{-1}v^{-1}+\frac{1}{2}v^{-2}\right]\\
=&\:8MI_0(u+1)^{-2} \left[\log(u+1)+\log\left(\frac{v-u-1}{v}\right)\right]+O((u+1)^{-2}).
\end{split}
\end{equation*}\\
\\
\underline{\textbf{Step 3:} Evaluating $J_4$}\\
\\
Now, let
\begin{equation*}
F_3(v;c)=-\frac{1}{2}v^{-2}\log\left(\frac{cv}{v-c}\right)+\frac{1}{2}c^{-2}\log\left(\frac{v}{v-c}\right)-\frac{1}{4}v^{-2}-\frac{1}{2}c^{-1}v^{-1}.
\end{equation*}
Then,
\begin{equation*}
\begin{split}
\frac{dF_3}{dv}(v;c)=&\:v^{-3}\log\left(\frac{cv}{v-c}\right)-\frac{1}{2}v^{-2}\cdot \left(\frac{v-c}{cv}\right)\cdot\left(-\frac{c^2}{(v-c)^2}\right)\\
&+\frac{1}{2}c^{-2}\frac{v-c}{v}\cdot\left(-\frac{c}{(v-c)^2}\right)+\frac{1}{2}v^{-3}+\frac{1}{2}c^{-1}v^{-2}\\
&=\:v^{-3}\log\left(\frac{cv}{v-c}\right)+\frac{1}{2}(v-c)^{-1}v^{-3}c^{-1}(c^2-v^2+c(v-c)+v(v-c))\\
&=\:v^{-3}\log\left(\frac{cv}{v-c}\right)
\end{split}
\end{equation*}
and we can write
\begin{equation*}
\begin{split}
J_4(u,v)=&\: 8MI_0\int_{u+2R_*}^{v}v'^{-3} \log \left(\frac{v'(u+1)}{v'-(u+1)}\right)\,dv'=8MI_0 F_3(v';u+1)|^{v'=v}_{v'=u+2R_*}\\
=&\:8MI_0\left[\frac{1}{2}(u+2R_*)^{-2}\log\left(\frac{(u+1)(u+2R_*)}{2R_*-1}\right)-\frac{1}{2}(u+1)^{-2}\log\left(\frac{u+2R_*}{2R_*-1}\right)\right]\\
&8MI_0\left[-\frac{1}{2}v^{-2}\log\left(\frac{(u+1)v}{v-u-1}\right)+\frac{1}{2}(u+1)^{-2}\log\left(\frac{v}{v-u-1}\right)\right]\\
&+O((u+1)^{-2})\\
=&\:4MI_0[(u+1)^{-2}-v^{-2}]\log(u+1)-4MI_0[(u+1)^{-2}-v^{-2}]\log\left(\frac{v-u-1}{v}\right)\\
&+O((u+1)^{-2}).
\end{split}
\end{equation*}

Hence,
\begin{equation*}
\begin{split}
J_3(u,v)+J_4(u,v)=&\:12MI_0(u+1)^{-2}\log(u+1)-4MI_0v^{-2}\log(u+1)\\
&+4MI_0[(u+1)^{-2}+v^{-2}]\log\left(\frac{v-u-1}{v}\right).
\end{split}
\end{equation*}

By combining the above estimates, we arrive at \eqref{eq:asymphiinfty}.
\end{proof}

\subsection{Estimates for $T(r\psi )$ in the case $I_0 [\psi ] \neq 0$}
\begin{proposition}For all spherically symmetric solutions $\psi$ to the wave equation \eqref{we} on the $(\mathcal{M},g)$ backgrounds of Section \ref{preliminaries} with non-vanishing Newman--Penrose constant $I_0$  there  exists a constant $C_R>0$ such that we can estimate in $\{r\geq R\}\cap\{u\geq 0\}$:
\begin{equation}
\label{eq:2ndasympTphi}
\begin{split}
\Bigg|&T(r\psi)(u,v)+2I_0[(u+1)^{-2}-v^{-2}]-8MI_0(u+1)^{-3}\log(u+1)-16MI_0v^{-3}\log v\\
&+8MI_0(u+1)(v^{-3}-(u+1)^{-3})(v-u-1)^{-1}+8MI_0(u+1)^{-3}\log \left(\frac{v}{v-(u+1)}\right)\\
&+8MI_0v^{-3}\log \left(\frac{v(u+1)}{v-(u+1)}\right)\Bigg|\\
\leq&\: C_R\left(I_0+P_T+\sqrt{E_T}\right)(u+1)^{-3}.
\end{split}
\end{equation}
\label{propot}
\end{proposition}
\begin{proof}
Let $(u,v)\in \{r\geq R\}\cap\{u\geq 0\}$ and denote $v_R(u)=u+2R_*$ and $\phi:=r\psi$. Then we can apply the fundamental theorem of calculus in $v$ together with the identity $T=\partial_u+\partial_v$ to obtain:
\begin{equation*}
\begin{split}
T\phi(u,v)=&\:T\phi(u,v_R(u))+\int_{v_R(u)}^v \partial_vT\phi(u,v')\,dv'\\
=&\:T\phi(u,v_R(u))+\int_{v_R(u)}^v \partial_v^2\phi(u,v')\,dv'+\int_{v_R(u)}^v \partial_u\partial_v\phi(u,v')\,dv'\\
=&\:T\phi(u,v_R(u))+\partial_v\phi(u,v)-\partial_v\phi(u,v_R(u))+\int_{v_R(u)}^v \partial_u\partial_v\phi(u,v')\,dv'.
\end{split}
\end{equation*}
By applying moreover the fundamental theorem of calculus in $u$, we can rewrite the above equation to obtain:
\begin{equation*}
\begin{split}
T\phi(u,v)=&\:T\phi(u,v_R(u))+\partial_v\phi(u,v)-\partial_v\phi(u_{\gamma}(v_R(u)),v_R(u))-\int_{u_{\gamma}(v_R(u))}^u \partial_u\partial_v\phi(u',v_R(u))\,du'\\
&+\int_{v_R(u)}^v \partial_u\partial_v\phi(u,v')\,dv',
\end{split}
\end{equation*}
where $u_{\gamma}(v_R(u))=\frac{1}{2}v_R(u)-R_*=\frac{u}{2}$.

Note that $r_*(\frac{u}{2},2R_*+u)=\frac{u}{4}+R_*$. We now split the $v$ integral as follows (see also Figure \ref{fig:rgeqR2}):
\begin{equation*}
\begin{split}
\int_{v_R(u)}^v \partial_u\partial_v\phi(u,v')\,dv'=&\:\int_{v_R(u)}^{u+2(\frac{u}{4}+R_*)} \partial_u\partial_v\phi(u,v')\,dv'+\int^{v}_{u+2(\frac{u}{4}+R_*)} \partial_u\partial_v\phi(u,v')\,dv'\\
&\:=\int_{2R_*+u}^{2R_*+\frac{3}{2}u} \partial_u\partial_v\phi(u,v')\,dv'+\int^{v}_{2R_*+\frac{3}{2}u} \partial_u\partial_v\phi(u,v')\,dv'.
\end{split}
\end{equation*}

We can then write:
\begin{equation*}
\begin{split}
T\phi(u,v)=&\:T\phi(u,v_R(u))+\partial_v\phi(u,v)-\partial_v\phi\left(\frac{u}{2},u+2R_*\right)+\int^{v}_{\frac{3}{2}u+2R_*} \partial_u\partial_v\phi(u,v')\,dv'\\
&+\int_{2R_*+u}^{\frac{3}{2}u+2R_*} \partial_u\partial_v\phi(u,v')\,dv'-\int_{\frac{u}{2}}^u \partial_u\partial_v\phi(u',v_R(u))\,du'.
\end{split}
\end{equation*}

We further decompose:
\begin{equation*}
\begin{split}
\int^{v}_{\frac{3}{2}u+2R_*} \partial_u\partial_v\phi(u,v')\,dv'=&\:\int^{v}_{\frac{3}{2}u+2R_*} -8MI_0(v'-u-1)^{-2}v'^{-1}(u+1)^{-1}(u,v')\,dv'\\
&+\int^{v}_{\frac{3}{2}u+2R_*}  \left[-\frac{1}{4}DD'\psi(u,v')+8MI_0(v'-u-1)^{-2}v'^{-1}(u+1)^{-1}\right]\,dv'
\end{split}
\end{equation*}
and
\begin{equation*}
\begin{split}
\int_{2R_*+u}^{\frac{3}{2}u+2R_*} &\partial_u\partial_v\phi(u,v')\,dv'-\int_{\frac{u}{2}}^u \partial_u\partial_v\phi(u',v_R(u))\,du'\\
=&\:\int_{2R_*+u}^{\frac{3}{2}u+2R_*} -8MI_0(v'-u-1)^{-2}v'^{-1}(u+1)^{-1}\,dv'\\
&-\int_{\frac{u}{2}}^u -8MI_0(v_R(u)-u'-1)^{-2}v_R(u)^{-1}(u'+1)^{-1}\,du'\\
&+\int_{2R_*+u}^{2R_*+\frac{3}{2}u} \left[-\frac{1}{4}DD'\psi(u,v')+8MI_0(v'-u-1)^{-2}v'^{-1}(u+1)^{-1}\right]\,dv'\\
&-\int_{\frac{u}{2}}^u \left[-\frac{1}{4}DD'\psi(u',v_R(u))+8MI_0(v_R(u)-u'-1)^{-2}v_R(u)^{-1}(u'+1)^{-1}\right]\,du'.
\end{split}
\end{equation*}

By \eqref{eq:loasymppsi}  we have that
\begin{equation*}
\begin{split}
\Bigg| &-\frac{1}{4}DD'\psi+8MI_0(v-u-1)^{-2}v^{-1}(u+1)^{-1}\Bigg|\\
\leq&\: C(I_0+P+\sqrt{E})\left[v^{-1}(u+1)^{-1-\epsilon}(v-u-1)^{-2}+v^{-1}(u+1)^{-1}(v-u-1)^{-3+\max\{\eta,1-\beta\}}\right].
\end{split}
\end{equation*}
and therefore we can estimate
\begin{equation*}
\begin{split}
\Bigg|&\int^{v}_{\frac{3}{2}u+2R_*}  \left[-\frac{1}{4}DD'\psi(u,v')+8MI_0(v'-u-1)^{-2}v'^{-1}(u+1)^{-1}\right]\,dv'\Bigg|\\
\leq&\: C(I_0+P+\sqrt{E})\left[(u+1)^{-2-\epsilon}(v-u-1)^{-1}+(u+1)^{-2}(v-u-1)^{-2+\max\{\eta,1-\beta\}}+(1+u)^{-3}\right].
\end{split}
\end{equation*}
Similarly, by \eqref{eq:loasympTpsi}  we have that
\begin{equation}
\label{eq:mainestTpsi}
\begin{split}
\Bigg| &T\left(-\frac{1}{4}DD'\psi+8MI_0(v-u-1)^{-2}v^{-1}(u+1)^{-1}\right)\Bigg|\\
\leq&\: C(I_0+P_T+\sqrt{E_T})[v^{-1}(u+1)^{-2-\epsilon}(v-u-1)^{-2}\\
&+(v^{-1}(u+1)^{-2}+v^{-2}(u+1)^{-1})(v-u-1)^{-3+\max\{\eta,1-\beta\}}].
\end{split}
\end{equation}
Observe that
\begin{equation}
\label{eq:spacetimeint}
\begin{split}
\int_{2R_*+u}^{2R_*+\frac{3}{2}u}& f(u,v')\,dv'-\int_{\frac{1}{2}u}^u f(u',v_R(u))\,du'\\
=&\:2\int_{R_*}^{\frac{u}{4}+R_*}f(t',r_*')|_{t=u+r_*'}\,dr_*'-\int_{R_*}^{\frac{u}{4}+R_*}f(t',r_*')|_{t=u+2R_*-r_*'}\,dr_*'\\
=&\:2\int_{R_*}^{\frac{u}{4}+R_*}\int^{u+r_*}_{u+2R_*-r_*}T(f)(t',r'_*)\,dt'dr'_*,
\end{split}
\end{equation}
see Figure \ref{fig:rgeqR2}  below.
\begin{figure}[h!]
\begin{center}
\includegraphics[width=3in]{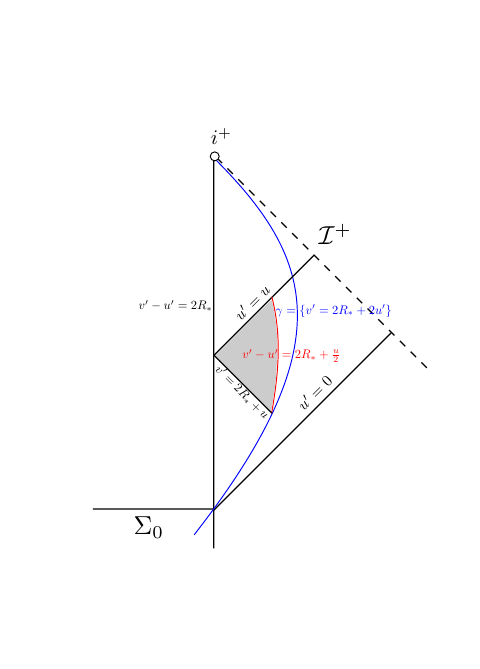}
\caption{\label{fig:rgeqR2} The shaded region depicts the integration region appearing on the very left-hand side of \eqref{eq:spacetimeint}.}
\end{center}
\end{figure}

We can perform a change of variables:
\begin{equation*}
\begin{split}
2\int_{R_*}^{\frac{u}{4}+R_*}\int^{u+r_*}_{u+2R_*-r_*} T(f)(t'-r_*',t'+r_*')\,dt'dr_*'=&\:\int_{2R_*+u}^{2R_*+\frac{3}{2}u}\int^{u}_{v'-2R_*-\frac{u}{2}} T(f)(u',v')\,du'dv'.
\end{split}
\end{equation*}

We now choose
\begin {equation*}
f(u',v')=-\frac{1}{4}DD'\psi(u',v')+8MI_0(v'-u'-1)^{-2}v'^{-1}(u'+1)^{-1}.
\end{equation*}

and apply the above integral equalities to obtain
\begin{equation*}
\begin{split}
\Bigg|\int_{2R_*+u}^{2R_*+\frac{3}{2}u}& \left[-\frac{1}{4}DD'\psi(u,v')+8MI_0(v'-u-1)^{-2}v'^{-1}(u+1)^{-1}\right]\,dv'\\
&-\int_{\frac{u}{2}}^u \left[-\frac{1}{4}DD'\psi(u',v_R(u))+8MI_0(v_R(u)-u'-1)^{-2}v_R(u)^{-1}(u'+1)^{-1}\right]\,du'\Bigg|\\
=&\:\left|\int_{2R_*+u}^{2R_*+\frac{3}{2}u}\int^{u}_{v'-2R_*-\frac{u}{2}} T\left(-\frac{1}{4}DD'\psi+8MI_0(v-u-1)^{-2}v^{-1}(u+1)^{-1}\right)(u',v')\,du'dv'\right|\\
\leq&\:C(I_0+P_T+\sqrt{E_T})\int_{2R_*+u}^{2R_*+\frac{3}{2}u}\int^{u}_{v'-2R_*-\frac{u}{2}}v'^{-1}(u'+1)^{-2-\epsilon}(v'-u'-1)^{-2}\,du'dv'\\
&+C(I_0+P_T+\sqrt{E_T})\int_{2R_*+u}^{2R_*+\frac{3}{2}u}\int^{u}_{v'-2R_*-\frac{u}{2}}(v'^{-1}(u'+1)^{-2}\\
&+v'^{-2}(u'+1)^{-1})(v'-u'-1)^{-3+\max\{\eta,1-\beta\}}\,du'dv'\\
\leq&\:C(I_0+P_T+\sqrt{E_T})\int_{2R_*+u}^{2R_*+\frac{3}{2}u}(v'+1)^{-3-\epsilon}\int^{u}_{v'-2R_*-\frac{u}{2}}(v'-u'-1)^{-2}\,du'dv'\\
&+C(I_0+P_T+\sqrt{E_T})\\
&\cdot \int_{2R_*+u}^{2R_*+\frac{3}{2}u}(v'+1)^{-3}\int^{u}_{v'-2R_*-\frac{u}{2}}(v'-u'-1)^{-3+\max\{\eta,1-\beta\}}\,du'dv',
\end{split}
\end{equation*}
where we applied \eqref{eq:mainestTpsi} in the first inequality and used that $v'^{-1}\geq (2R_*+\frac{3}{2}u)^{-1}$ to obtain the second inequality.

Hence, it follows easily that
\begin{equation*}
\begin{split}
\Bigg|\int_{2R_*+u}^{2R_*+\frac{3}{2}u}& \left[-\frac{1}{4}DD'\psi(u,v')+8MI_0(v'-u-1)^{-2}v'^{-1}(u+1)^{-1}\right]\,dv'\\
&-\int_{\frac{u}{2}}^u \left[-\frac{1}{4}DD'\psi(u',v_R(u))+8MI_0(v_R(u)-u'-1)^{-2}v_R(u)^{-1}(u'+1)^{-1}\right]\,du'\Bigg|\\
\leq&\:C(I_0+P_T+\sqrt{E_T})(u+1)^{-3}.
\end{split}
\end{equation*}

We are left with computing the integrals
\begin{equation*}
J_5(u,v)=:\int_{2R_*+u}^{v} -8MI_0(v'-u-1)^{-2}v'^{-1}(u+1)^{-1}\,dv'
\end{equation*}
and
\begin{equation*}
J_6(u)=:\int_{\frac{u}{2}}^u 8MI_0(v_R(u)-u'-1)^{-2}v_R(u)^{-1}(u'+1)^{-1}\,du'.
\end{equation*}

Let us introduce the functions:
\begin{align*}
G_1(v;x)=&\:x^{-2}\log\left(\frac{v}{v-x}\right)-x^{-1}(v-x)^{-1},\\
G_2(x;c)=&\:c^{-2}\log\left(\frac{x}{c-x}\right)+c^{-1}(c-x)^{-1}.
\end{align*}
Note that
\begin{align*}
\frac{dG_1}{dv}(v;x)=&\:v^{-1}(v-x)^{-2},\\
\frac{dG_2}{dx}(x;c)=&\: x^{-1}(c-x)^{-2}.
\end{align*}

Hence,
\begin{align*}
J_5(u,v)=&-8MI_0(u+1)^{-1}G_1(v';u+1)\Big|^{v'=v}_{v'=2R_*+u}\\
=&-8MI_0(u+1)^{-2}\left[(2R_*-1)-(u+1)^{-1}\log(u+1)\right]+O((u+1)^{-3})\\
&-8MI_0(u+1)^{-3}\log\left(\frac{v}{v-u-1}\right)+8MI_0(u+1)^{-2}(v-u-1)^{-1},\\
J_6(u)=&\:8MI_0(u+2R_*)^{-1}G_2(x;u+2R_*)\Big|^{x=u+1}_{x=\frac{u}{2}+1}\\
=&\:8MI_0(u+2R_*)^{-1}\left[(u+2R_*)^{-1}(2R_*-1)+(u+1)^{-2}\log(u+1)\right]+O((u+1)^{-3})
\end{align*}
and we can conclude that
\begin{equation*}
\begin{split}
J_5(u)+J_6(u)=&\:16MI_0(u+1)^{-3}\log(u+1)-8MI_0(u+1)^{-3}\log\left(\frac{v}{v-u-1}\right)\\
&+8MI_0(u+1)^{-2}(v-u-1)^{-1}+O((u+1)^{-3}).
\end{split}
\end{equation*}
By combining the above estimates and using that $T\phi(u,v_R(u))\leq C_R(I_0+P_T+\sqrt{E_T})(u+1)^{-3}$, we therefore obtain:
\begin{equation}
\begin{split}
\label{eq:auxasymtphi}
\Big|&T\phi(u,v)-\partial_v\phi(u,v)+\partial_v\phi\left(\frac{u}{2},u+2R_*\right)-16MI_0(u+1)^{-3}\log(u+1)\\
&+8MI_0(u+1)^{-3}\log\left(\frac{v}{v-u-1}\right)-8MI_0(u+1)^{-2}(v-u-1)^{-1}\Big|\\
\leq&\:C_R(I_0+P_T+\sqrt{E_T})(u+1)^{-3}.
\end{split}
\end{equation}
By \eqref{eq:2ndoasympdvphi} with $v$ replaced by $u+2R_*$ and $u$ replaced by $\frac{u}{2}$, we have that
\begin{equation*}
\left|\partial_v\phi\left(\frac{u}{2},u+2R_*\right)-2I_0(u+1)^{-2}-8MI_0(u+1)^{-3}\log(u+1)\right|\leq C_R(I_0+P+\sqrt{E})(u+1)^{-3},
\end{equation*}
so we can conclude after using once more \eqref{eq:2ndoasympdvphi}, together with \eqref{eq:auxasymtphi}:
\begin{equation*}
\begin{split}
\Bigg|&T\phi(u,v)+2I_0[(u+1)^{-2}-v^{-2}]-8MI_0(u+1)^{-3}\log(u+1)-16MI_0v^{-3}\log v\\
&+8MI_0(u+1)(v^{-3}-(u+1)^{-3})(v-u-1)^{-1}+8MI_0(u+1)^{-3}\log \left(\frac{v}{v-(u+1)}\right)\\
&+8MI_0v^{-3}\log \left(\frac{v(u+1)}{v-(u+1)}\right)\Bigg|\\
\leq&\: C_R\left(I_0+P_T+\sqrt{E_T}\right)(u+1)^{-3}.
\end{split}
\end{equation*}
\end{proof}

\subsection{Estimates for $r\psi$ in the case $I_0 [ \psi ] = 0$}

\begin{proposition}For all spherically symmetric solutions $\psi$ to the wave equation \eqref{we} on the $(\mathcal{M},g)$ backgrounds of Section \ref{preliminaries} with vanishing Newman--Penrose constant $I_0$  there  exists a constant $C_R>0$ such that we can estimate in $\{r\geq R\}\cap\{u\geq 0\}$:
\begin{equation}
\label{eq:2ndasympphiNP0}
\begin{split}
\Bigg|&\phi(u,v)+2I_0^{(1)}[(u+1)^{-2}-v^{-2}]-8MI_0^{(1)}(u+1)^{-3}\log(u+1)-16MI_0^{(1)}v^{-3}\log v\\
&+8MI_0^{(1)}(u+1)(v^{-3}-(u+1)^{-3})(v-u-1)^{-1}+8MI_0^{(1)}(u+1)^{-3}\log \left(\frac{v}{v-(u+1)}\right)\\
&+8MI_0^{(1)}v^{-3}\log \left(\frac{v(u+1)}{v-(u+1)}\right)\Bigg|\\
\leq&\: C_R\left(I_0^{(1)}+P_T[\psi^{(1)}]+\sqrt{E_T}[\psi^{(1)}]\right)(u+1)^{-3}.
\end{split}
\end{equation}
\label{propovic}
\end{proposition}
\begin{proof}
In view of Section \ref{timeinverted} there is a unique time-integral $\psi^{(1)}$ associated to $\psi$. Given the assumption \eqref{dbehavior} we have from the results of \cite{paper2} that
\begin{equation*}
\partial_r(r\psi^{(1)})|_{\{u=0\}}=I_0^{(1)}r^{-2}+O(r^{-3}),
\end{equation*}
which implies that \eqref{eq:rasyminitdata} holds for $\psi^{(1)}$. Hence,  the estimate of Proposition \ref{propot} applies with $\psi^{(1)}$ replacing $\psi$ and $I^{(1)}_{0}$ replacing $I_0$. The results follows from the fact that $T(r\psi^{(1)})=\phi$.
\end{proof}

\section{Proof of the main theorems}\label{pmain}

We decompose $\psi$ as follows
\[\psi=\psi_0+\psi_{\geq 1},\]
where 
\[\psi_0=\frac{1}{4\pi}\int_{\mathbb{S}^2}\psi\,d\omega\]
and
\[\psi_{\geq 1}=\psi-\psi_0.\]

\subsection{Proof of Theorem \ref{thm2}}\label{pthm2}
\begin{proof}
In view of Proposition 4.3 of \cite{paper2} we have the following estimate for the radiation field of $\psi_{\geq 1}$:
\[ |r\psi_{\geq 1}|_{\mathcal{I}^{+}}(u,\cdot)\leq C\cdot \frac{\sqrt{E}}{(u+1)^{\frac{5}{2}-\epsilon}},\]
for some $\epsilon<0.25$. The result for the spherical mean $\psi_0$ is a corollary of Proposition \ref{prop:asymphi} after fixing $u$ and  taking the limit as $v\rightarrow \infty$ and observing that the limit of the expression 
\[2I_0v^{-1}-4MI_0v^{-2}\log(u+1)+8MI_0v^{-2}\log v+4MI_0[(u+1)^{-2}+v^{-2}]\log\left(\frac{v-u-1}{v}\right)\]
on the right hand side of \eqref{eq:asymphi} vanishes. 

The result follows by adding the estimates for $\psi_0$ and $\psi_{\geq 1}$.

\end{proof}
\subsection{Proof of Theorem \ref{thm3}}\label{pthm3}
\begin{proof}
In view of Proposition 4.3 of \cite{paper2} we have the following estimate for the radiation field of $T\psi_{\geq 1}$
\[ |T(r\psi)_{\geq 1}|_{\mathcal{I}^{+}}(u,\cdot)\leq C\cdot \frac{\sqrt{E}}{(u+1)^{\frac{7}{2}-\epsilon}},\]
for some $\epsilon<0.25$.

The result for  $T\psi_0$  a corollary of Proposition \ref{propot} after fixing $u$ and  taking the limit as $v\rightarrow \infty$ and observing that the limit of the expression
\begin{equation*}
\begin{split}
&-2I_0v^{-2}-16MI_0v^{-3}\log v+8MI_0(u+1)(v^{-3}-(u+1)^{-3})(v-u-1)^{-1}\\&+8MI_0(u+1)^{-3}\log \left(\frac{v}{v-(u+1)}\right)+8MI_0v^{-3}\log \left(\frac{v(u+1)}{v-(u+1)}\right)
\end{split}
\end{equation*}
on the right hand side of \eqref{eq:2ndasympTphi} vanishes.

The result follows by adding the estimates for $\psi_0$ and $\psi_{\geq 1}$.

\end{proof}

\subsection{Proof of Theorem \ref{thm1}}\label{pthm1}
\begin{proof} The compact support of the initial data guarantees the vanishing of the Newman--Penrose constant and consequently the existence of the associated time-integral $\psi^{(1)}$. Hence Proposition \ref{propovic} applies. The result follows after fixing $u$ and taking the limit as $v\rightarrow \infty$ of \eqref{eq:2ndasympphiNP0} and observing that the limit of the expression 
\begin{equation*}
\begin{split}
-&2I_0^{(1)}v^{-2}-16MI_0^{(1)}v^{-3}\log v+8MI_0^{(1)}(u+1)(v^{-3}-(u+1)^{-3})(v-u-1)^{-1}\\&+8MI_0^{(1)}(u+1)^{-3}\log \left(\frac{v}{v-(u+1)}\right)+8MI_0^{(1)}v^{-3}\log \left(\frac{v(u+1)}{v-(u+1)}\right)\
\end{split}
\end{equation*}
on the right hand side of \eqref{eq:2ndasympphiNP0} vanishes.

\end{proof}

\section{Acknowledgments}

The second author (S.A.) acknowledges support through NSF grant DMS-1265538, NSERC grant 502581, an Alfred P. Sloan Fellowship in Mathematics and the Connaught Fellowship 503071.

\appendix

\section{Energy norms}
\label{energynorms}
In this appendix we define the norms used in the Theorem \ref{thm:loasymp} and in the main theorems \ref{thm1}, \ref{thm2}, \ref{thm3}. 

We first define the following weighted $L^{\infty}$ norms in $(u,r)$ coordinates: 
\begin{equation}
\label{np1-der1}
 P [\psi] := \left\|Dr^{3}\cdot \left(\partial_r\phi_0-\frac{I_0 [\psi ]}{r^2}\right)\right\|_{L^{\infty}(\Sigma_{0})},
\end{equation}
and
\begin{equation}
\label{np1-der2}
 P_{T}  [\psi]:=\left\| Dr^4\cdot\partial_{r}\left( D\partial_{r}\phi_0-D\frac{I_{0}[\psi]}{r^{2}}  \right)  \right\|_{L^{\infty}(\Sigma_{0})} .
\end{equation}

Any suitably regular function $f$ admits the following decomposition in angular frequencies in $\mathcal{M}$
\begin{equation*}
f (u,v,\theta,\varphi)=\sum_{\ell'=0}^{\infty} \psi_{\ell=\ell'}(u,v,\theta,\varphi).
\end{equation*}
Also let $\Omega$  denote any of the three Killing vector fields $\Omega_i$, $i =1,2,3$ associated to the spherical symmetry of our spacetime given by
\begin{align*}
\Omega_1&=\sin \varphi \partial_{\theta}+\cot\theta \cos \varphi \partial_{\varphi},\\
\Omega_2&=-\cos \varphi \partial_{\theta}+\cot\theta \sin \varphi \partial_{\varphi},\\
\Omega_3&=\partial_{\varphi}.
\end{align*}
Let also $\Omega^{\alpha}$ denote any of the product of these vector fields:
\begin{equation*}
\Omega^{\alpha}=\Omega_1^{k_1}\Omega_2^{k_2}\Omega_3^{k_3},
\end{equation*}
where $|\alpha | =  k_1 + k_2 + k_3$ for $k_1 , k_2 , k_3 \in \mathbb{N}$.

The energy-momentum tensor is defined as follows
\begin{equation*}
\mathbf{T}_{\alpha \beta}[f]=\partial_{\alpha}f \partial_{\beta}f-\frac{1}{2}{g}_{\alpha \beta} (g^{-1})^{\kappa \lambda}\partial_{\kappa}f \partial_{\lambda} f,
\end{equation*}
and has the property that $\text{div} \mathbf{T} [ \psi ] = \Box_g\psi \cdot d\psi$. The energy current $J^V [f]$ is defined with respect to a function $f$ and two vector field $V_1$ and $V_2$ as
$$ J^{V_1} [f] \cdot V_2 := \mathbf{T} ( V_1 , V_2 ) . $$

We also recall the following definitions of \cite{paper1, paper2}:
\begin{align*}
\Phi &=r^2 \partial_r ( r\psi ) ,\\
\widetilde{\Phi} &=r(r-M) \partial_r (r\psi) ,\\
\Phi_{(2)} &=r^2 \partial_r \left( r^2 \partial_r (r\psi) \right) .
\end{align*}
In our case we decompose a linear wave as follows:
$$ \psi = \psi_0 + \psi_{\ell = 1} + \psi_{\ell \geq 2} . $$
The timelike vector field $N$ is as defined in Section 7.2 of \cite{paper2}.
Next we recall as well from \cite{paper2} the definition of the following energy norms
\begin{equation*}
\begin{split}
E^{\epsilon}_{0,I_0\neq0;k}[\psi]=&\:\sum_{l\leq 3+3k}\int_{\Sigma_{0}}J^N[T^l\psi]\cdot n_{0}\,d\mu_{\Sigma_0}\\
&+\sum_{l\leq 2k}\int_{\mathcal{N}_{0}} r^{3-\epsilon}(\partial_rT^l\phi)^2\,dr+r^{2}(\partial_rT^{l+1}\phi)^2+r(\partial_rT^{2+l}\phi)^2\,dr\\
&+\sum_{\substack{m\leq k\\ l\leq 2k-2m+\min\{k,1\}}} \int_{\mathcal{N}_{0}}r^{2+2m-\epsilon}(\partial_r^{1+m}T^{l}\phi)^2\, dr\\
&+\int_{\mathcal{N}_{0}}r^{3+2k-\epsilon}(\partial_r^{1+k}\phi)^2\, dr ,
\end{split}
\end{equation*}
and
\begin{equation*}
\widetilde{E}^{\epsilon}_{0,I_0\neq 0;k+1}[\psi]=E^{\epsilon}_{0,I_0\neq 0;k+1}[\psi]+\sum_{j=0}^k\int_{\Sigma_0}J^N[NT^j\psi]\cdot n_0\,d\mu_0,
\end{equation*}
for spherically symmetric linear waves $\psi$, the norms
\begin{equation*}
\begin{split}
E_{1;k}^{\epsilon}[\psi]\doteq &\:\sum_{\substack{l\leq 6+3k}}\int_{\Sigma_{0}}J^N[T^l\psi]\cdot n_{0}\;d\mu_{\Sigma_0}\\
&+\sum_{ l\leq 4+2k}\int_{\mathcal{N}_0}r^{2}(\partial_rT^l\phi)^2+r^{1}(\partial_rT^{1+l}\phi)^2\,d\omega dr\\
&+\sum_{l\leq 3,m\leq 2k}\int_{\mathcal{N}_{0}}r^{4-l-\epsilon}(\partial_rT^{l+m}\widetilde{\Phi})^2\,d\omega dr\\
&+\sum_{\substack{ m\leq \max\{k-1,0\}\\ l\leq k-2m+\min\{k,1\}}}\int_{\mathcal{N}_{0}}r^{4+2m-\epsilon}(\partial_r^{1+m}T^{l}\widetilde{\Phi})^2\,d\omega dr\\
&+\sum_{\substack{ m\leq k\\ l\leq 2k-2m+1}}\int_{\mathcal{N}_{0}} r^{3+2m-\epsilon}(\partial_r^{1+m}T^{l}\widetilde{\Phi})^2\,d\omega dr\\
&+\int_{\mathcal{N}_{0}} r^{4+2k-\epsilon}(\partial_r^{1+k}\widetilde{\Phi})^2\,d\omega dr , 
\end{split}
\end{equation*}
for linear waves localized at angular frequency $\ell = 1$, and the norms
\begin{equation*}
\begin{split}
E_{2;k}^{\epsilon}[\psi]\doteq &\:\sum_{\substack{|\alpha|\leq k\\ l+|\alpha|\leq 6+3k}}\int_{\Sigma_{0}}J^N[T^l\Omega^{\alpha}\psi]\cdot n_{0}\;d\mu_{\Sigma_0}\\
&+\sum_{ l\leq 4+2k}\int_{\mathcal{N}_0}r^{2}(\partial_rT^l\phi)^2+r^{1}(\partial_rT^{1+l}\phi)^2\,d\omega dr\\
&+\sum_{ l\leq 2k+2}\int_{\mathcal{N}_{0}} r^{2-\epsilon}(\partial_rT^{l}{\Phi})^2+r^{1-\epsilon}(\partial_rT^{l+1}{\Phi})^2\,d\omega dr\\
&+\sum_{\substack{|\alpha|\leq k\\l+|\alpha|\leq 2k}}\int_{\mathcal{N}_{0}} r^{2-\epsilon}(\partial_rT^{l}\Omega^{\alpha}{\Phi}_{(2)})^2+r^{1-\epsilon}(\partial_rT^{l+1}\Omega^{\alpha}{\Phi}_{(2)})^2\,d\omega dr\\
&+\sum_{\substack{|\alpha|\leq \max\{0,k-1\}\\ m\leq \max\{k-1,0\}\\ l+|\alpha|\leq k-2m+\min\{k,1\}}}\int_{\mathcal{N}_{0}}r^{2+2m-\epsilon}(\partial_r^{1+m}\Omega^{\alpha}T^{l}{\Phi}_{(2)})^2\,d\omega dr\\
&+\sum_{\substack{|\alpha|\leq \max\{0,k-1\}, m\leq k\\ l+|\alpha|\leq 2k-2m+1}}\int_{\mathcal{N}_{0}} r^{1+2m-\epsilon}(\partial_r^{1+m}\Omega^{\alpha}T^{l}{\Phi}_{(2)})^2\,d\omega dr\\
&+\int_{\mathcal{N}_{0}} r^{2+2k-\epsilon}(\partial_r^{1+k}{\Phi}_{(2)})^2\,d\omega dr.
\end{split}
\end{equation*}
for linear waves localized at angular frequencies $ \ell \geq 1$. 

Finally we define the norms used in Theorem \ref{thm:loasymp}:
\begin{align}\label{Epsi}
E [ \psi ] := \widetilde{E}^{\epsilon}_{0, I_0 \neq 0 ; 1} [\psi_0 ] + E_{1;1}^{\epsilon} [ \psi_{\ell = 1} ] + E_{2;1}^{\epsilon} [ \psi_{\ell \geq 2} ] , \\
E_T [\psi ] :=  \widetilde{E}^{\epsilon}_{0, I_0 \neq 0 ; 2} [\psi_0 ] + E_{1;2}^{\epsilon} [ \psi_{\ell = 1} ] + E_{2;2}^{\epsilon} [ \psi_{\ell \geq 2} ]  .
\end{align}


\end{document}